\documentclass[sigconf,10pt]{acmart}

\copyrightyear{2020}
\acmYear{2020}
\setcopyright{acmcopyright}
\acmConference[SIGMOD'20]{Proceedings of the 2020 ACM SIGMOD International Conference on Management of Data}{June 14--19, 2020}{Portland, OR, USA}
\acmBooktitle{Proceedings of the 2020 ACM SIGMOD International Conference on Management of Data (SIGMOD'20), June 14--19, 2020, Portland, OR, USA}
\acmPrice{15.00}
\acmDOI{10.1145/3318464.3389708}
\acmISBN{978-1-4503-6735-6/20/06}

\usepackage{balance}
\usepackage{graphicx}
\usepackage{url}
\usepackage[titlenumbered,ruled,noend,linesnumbered]{algorithm2e}
\usepackage{amsmath}
\usepackage{amsthm}
\usepackage{comment}
\usepackage{multirow}
\usepackage{hhline}
\usepackage{xcolor}
\usepackage{enumitem}

\DontPrintSemicolon 

\newcommand{\squishlist}{
	\begin{list}{$\bullet$}
		{ \setlength{\itemsep}{0pt}      \setlength{\parsep}{3pt}
			\setlength{\topsep}{3pt}       \setlength{\partopsep}{0pt}
			\setlength{\leftmargin}{1.5em} \setlength{\labelwidth}{1em}
			\setlength{\labelsep}{0.5em} } }

\newcommand{\squishlisttwo}{
\begin{list}{$\bullet$}
	{ \setlength{\itemsep}{0pt}    \setlength{\parsep}{0pt}
		\setlength{\topsep}{0pt}     \setlength{\partopsep}{0pt}
		\setlength{\leftmargin}{2em} \setlength{\labelwidth}{1.5em}
		\setlength{\labelsep}{0.5em} } }
\newcommand{\squishend}{
\end{list}  }

\def\script#1{\mathcal{#1}}

\def\mS{\script{S}}

\def\mC{\script{C}}

\def\mR{\script{R}}

\def\mI{\script{I}}

\def\mD{\script{D}}

{\hspace*{\fill}$\Box$\par}

\newtheorem{theorem}{Theorem}[section]

\newtheorem{proposition}[theorem]{Proposition}

\newtheorem{example}[theorem]{Example}
\newtheorem{definition}[theorem]{Definition}

\begin{document}
		
	\fancyhead{}
		
	\setlength{\abovedisplayskip}{3pt}%
	\setlength{\belowdisplayskip}{3pt}%
	\setlength{\abovedisplayshortskip}{0pt}%
	\setlength{\belowdisplayshortskip}{0pt}%
		
	\title{Learning Over Dirty Data Without Cleaning}

	\author{Jose Picado}
	\affiliation{%
		\institution{Oregon State University}
	}
	\email{jpicado@gmail.com}
	
	\author{John Davis}
	\affiliation{%
		\institution{Oregon State University}
	}
	\email{davisjo5@oregonstate.edu}
	
	\author{Arash Termehchy}
	\affiliation{%
		\institution{Oregon State University}
	}
	\email{termehca@oregonstate.edu}
	
	\author{Ga Young Lee}
	\affiliation{%
		\institution{Oregon State University}
	}
	\email{leegay@oregonstate.edu}

	\begin{abstract}
			Real-world datasets are dirty and contain many errors. 
			Examples of these issues are violations of integrity constraints, duplicates, and inconsistencies in representing data values and entities. Learning over dirty databases may result in inaccurate models.
			Users have to spend a great deal of time and effort to repair data errors and create a clean database 
			for learning. Moreover, as the information required to repair these errors is not often available,
			there may be numerous possible clean versions for a dirty database.
			We propose {\it DLearn}, a novel relational learning system that learns directly over dirty databases effectively and efficiently without any preprocessing. 
			DLearn leverages database constraints to learn accurate relational models over inconsistent and heterogeneous data. 
			Its learned models represent patterns over all possible clean instances of the data in a usable form. 
			Our empirical study indicates that DLearn learns accurate models over large real-world databases efficiently.


	\end{abstract}
	
	\maketitle	
		
		\section{Introduction}
\label{section:heterolearn-introduction}
Users often would like to learn interesting relationships over relational databases \cite{StarAIW,SRLNIPS,10.1145/3183713.3199515,DeRaedt:2010:LRL:1952055,Getoor:SRLBook,QuickFOIL}. Consider the IMDb database ({\it imdb.com}) that contains information about movies whose schema fragments are shown in Table~\ref{table:schemas} (top). Given this database and some training examples, 
a user may want to learn a new relation {\it highGrossing(title)}, which indicates that the movie with a given title is high grossing.
Given a relational database and training examples for a new relation, 
{\it relational machine learning} (relational learning) algorithms learn (approximate) relational models and definitions of the target relation in terms of existing relations in the database~\cite{DeRaedt:2010:LRL:1952055,Getoor:SRLBook,Richardson:2006:MLN,Mihalkova:ICML:07,progolem,Quinlan:FOIL}. 
For instance, the user may provide a set of high grossing movies as positive examples and a set of low grossing movies as negative examples to a relational learning algorithm. Given the IMDb database and these examples, the algorithm may learn:
\begin{align*}
	\mathit{highGrossing}(x) \leftarrow & \mathit{movies}(y, x, z), \mathit{mov2genres}(y,\mathit{`comedy\text{'}}), \\
	 & \mathit{mov2releasedate}(y,\mathit{`May\text{'}}, u),
\end{align*}
which indicates that high grossing movies are often released in May and their genre is {\it comedy}. 
One may assign weights to these definitions to describe their prevalence in the data according their training accuracy \cite{StarAIW,Richardson:2006:MLN}. 
As opposed to other machine learning algorithms, relational learning methods do not require the data points to be statistically  independent and follow the same identical distribution (IID) \cite{10.1145/3183713.3199515}.
Since a relational database usually contain information about multiple types of entities, the relationships between 
these entities often violate the IID assumption.  
Also, the data about each type of entities may follow a distinct 
distribution. 
This also holds if one wants to learn over the data gathered from multiple data sources as each data source may have a distinct data distribution.
Thus, using other learning methods on these databases results
in biased and inaccurate models \cite{StarAIW,SRLNIPS,10.1145/3183713.3199515}.
Since relational learning algorithms leverage the structure of the database directly to learn new relations, 
they do not need the tedious process of feature engineering.
In fact, they are used to discover features for the 
downstream non-relational models \cite{lao-etal-2015-learning}.
Thus, they have been widely used over relational data, e.g., building usable query interfaces \cite{Abouzied:PODS:13,Maier:VLDB:2015,Kalashnikov:2018:FFQ:3183713.3183727}, information extraction \cite{StarAIW,10.1145/3183713.3199515}, 
and entity resolution \cite{Evans2018LearningER}.

\begin{table}
	\centering
	\caption{ Schema fragments for the IMDb and BOM. }
	\vspace{-10pt}
	\begin{tabular} { l l}
		\hline
		\multicolumn{2}{c}{IMDb} \\
		\hline
		movies(id, title, year) &  mov2countries(id, name)\\
		mov2genres(id, name) & mov2releasedate(id, month, year) \\
		\hline
		\hline
		\multicolumn{2}{c}{BOM} \\
		\hline
		\multicolumn{2}{c}{mov2totalGross(title, gross)} \\
		\multicolumn{2}{c}{highBudgetMovies(title)} \\
		\hline
	\end{tabular}
	\label{table:schemas}
\vspace{-15pt}
\end{table}

Real-world databases often contain inconsistencies  \cite{Bertossi2011DataCA,Doan:2012:PDI:2401764,Fan2009ReasoningAR,Getoor:2013:ERB:2487575.2506179,Cong:Gao:2007,10.14778/1952376.1952378,DBLP:series/synthesis/2012Fan}, which may prevent the relational learning algorithms from finding an accurate definition.
In particular, the information in a domain is sometimes spread across several databases. 
For example, IMDb does {\it not} contain the information about the budget or total grossing of movies. This information is available in another database called Box Office Mojo (BOM) ({\it boxofficemojo.com}), for which schema fragments are shown in Table~\ref{table:schemas} (bottom). To learn an accurate definition for {\it highGrossing}, the user has to collect data from the BOM database. 
However, the same entity or value may be represented in various forms in the original databases, e.g., 
the {\it titles} of the same movie in IMDb and BOM have different formats, e.g., the title of the movie {\it Star Wars: Episode IV} is represented in IMDb as {\it Star Wars: Episode IV - 1977} and in BOM as {\it Star Wars - IV}.  
A single database may also contain these type of heterogeneity as a relation may have duplicate tuples for the same entity, e.g.,
duplicate tuples fo the same movie in BOM.
A database may have other types of inconsistencies that violate the integrity of the data. 
For example, a movie in IMDb may have two different production years \cite{Cong:Gao:2007,10.14778/1952376.1952378,DBLP:series/synthesis/2012Fan}.

Users have to resolve inconsistencies and learn over the repaired database, which is very difficult and time-consuming for large databases \cite{Doan:2012:PDI:2401764,Getoor:2013:ERB:2487575.2506179}. 
Repairing inconsistencies usually leads to numerous clean instances as the information about the correct fixes is not often
available \cite{Bertossi2011DataCA,Burdick:2016:DFL:2966276.2894748,Fan2009ReasoningAR}.
An entity may match and be a potential duplicate of multiple distinct entities in the database.
For example, title {\it Star Wars} may match both titles {\it Star Wars: Episode IV - 1977} and {\it Star Wars: Episode III - 2005}. 
Since we know that the {\it Star Wars: Episode IV - 1977} and {\it Star Wars: Episode III - 2005} refer to two different movies, 
the title {\it Star Wars} must be unified with only one of them. For each choice, the user ends up with a distinct database instance.
Since a large database may have many possible matches, the number of clean database instances will be enormous.
Similarly, it is not often clear how to resolve data integrity violations.
For instance, if a movie has multiple production years, one may not know which year is correct.
Due to the sheer number of volumes, it is not possible to generate and materialize all clean instances for a large dirty database \cite{DBLP:series/synthesis/2012Fan}. 
Cleaning systems usually produce a subset of all clean instances, e.g., the ones that differ minimally with the original data \cite{DBLP:series/synthesis/2012Fan}. This approach still generates many repaired databases \cite{Bertossi2011DataCA,10.14778/1952376.1952378,DBLP:series/synthesis/2012Fan}. 
It is also shown that these conditions may not produce the correct instances \cite{DBLP:journals/debu/Ilyas16}. 
Thus, the cleaning process may result in many instances where it is not clear which one to use for learning.
It takes a great deal of time for users to manage these instances and decide which one(s) to use for learning.
Most data scientists spend more than 80\% of their time on such cleaning tasks \cite{Krishnan:2016:AID:2882903.2899409}.

Some systems aim at producing a single probabilistic database that contain information about a subset of possible clean instances \cite{Rekatsinas2017HoloCleanHD}. 
These systems, however, do not address the problem of duplicates and value heterogeneities as 
they assume that there always is a reliable table, akin to a dictionary, which gives the 
unique value that should replace each potential duplicate in the database.
However, given that different values represent the same entity, it is {\it not} clear what should replace
the final value in the clean database, e.g., whether {\it Star War} represents {\it Star Wars: Episode IV - 1977} or {\it Star Wars: Episode III - 2005}. They also allow violations of integrity constraints to generate the final probabilistic database efficiently, which may lead to inconsistent repairs. Moreover, to restrict the set of clean instances, they require attributes to have finite domains 
that does not generally hold in practice.

We propose a novel learning method that learns directly over dirty databases without materializing its clean versions, thus, it substantially reduces the effort needed to 
learn over dirty. 
The properties of clean data are usually expressed using declarative data constraints, e.g., functional dependencies,  \cite{AliceBook,Abedjan:2015:PRD:2811716.2811766,DBLP:conf/sigmod/ChuIKW16,DBLP:conf/pods/Fan08,Fan2009ReasoningAR,DBLP:conf/kr/BahmaniBKL12,DBLP:series/synthesis/2012Fan,Burdick:2016:DFL:2966276.2894748,Rekatsinas2017HoloCleanHD,Galhardas:2001:DDC:645927.672042}.
Our system uses the declarative constraints during learning.
These constraints may be provided by users or discovered from the data using profiling techniques \cite{Abedjan:2015:PRD:2811716.2811766,MDedup:Koumarelas}.
Our contributions are as follows:

\squishlisttwo
\item 
We introduce and formalize the problem of learning over an inconsistent database (Section~\ref{section:approach-explanation}). 

\item We propose a novel relational learning algorithm called {\it DLearn} to learn over inconsistent data (Section~\ref{section:HDLEARN}). 

\item Every learning algorithm chooses the final result based on its coverage of the training data.
We propose an efficient method to compute the coverage of a definition directly over the heterogeneous database (Section~\ref{sec:generalization}). 

\item We provide an efficient implementation of DLearn over a relational database system   (Section~\ref{section:heterolearn-implementation-details}).
 
\item  We perform an extensive empirical study over real-world datasets and show that DLearn scales to and learns efficiently and effectively over large data.

\end{list}

		\section{Background}
\subsection{Relational Learning}
\label{section:background:learning}
In this section, we review the basic concepts of relational learning over databases without any heterogeneity  \cite{DeRaedt:2010:LRL:1952055,Getoor:SRLBook}. 
We fix two mutually exclusive sets of relation and attribute symbols. 
A database schema $\mS$ is a finite set of relation symbols $R_i$, $1 \leq i \leq n$. Each relation $R_i$ is associated with a set of attribute symbols denoted as $R_i(A_1,\ldots,A_m)$. We denote the domain of values for attribute $A$ as $dom(A)$. Each database instance $I$ of schema $\mS$ maps a finite set of tuples to every relation $R_i$ in $\mS$. Each tuple $t$ is a function that maps each attribute symbol in $R_i$ to a value from its domain. 
We denote the value of the set of attributes $X$ of tuple $t$ in the database $I$ by $t^{I}[X]$ or
$t[X]$ if $I$ is clear from the context. Also, when it is clear from the context, 
we refer to an instance of a relation $R$ simply as $R$.
An {\it atom} is a formula in the form of $R(u_1,$ $\ldots,$ $u_n)$,
where $R$ is a relation symbol and $u_1,$ $\ldots,$ $u_n$ are {\it terms}. Each term is either a variable or a constant, i.e., value.
A {\it ground atom} is an atom that only contains constants.
A {\it literal} is an atom, or the negation of an atom.
A {\it Horn clause} (clause for short) is a finite set of literals that contains exactly one positive literal.
A {\it ground clause} is a clause that only contains ground atoms.
Horn clauses are also called Datalog rules (without negation) or conjunctive queries.
A {\it Horn definition} is a set of Horn clauses with the same positive literal, i.e., non-recursive Datalog program or union of conjunctive queries.
Each literal in the body is {\it head-connected} if it has a variable shared with the head literal or another head-connected literal.

Relational learning algorithms learn first-order logic definitions from an input relational database and training examples.
Training examples $E$ are usually tuples of a single {\it target relation}, and express positive ($E^+$) or negative ($E^-$) examples.
The input relational database is also called {\it background knowledge}.
The {\it hypothesis space} is the set of all possible first-order logic definitions that the algorithm can explore. It is usually restricted to Horn definitions to keep learning efficient. 
Each member of the hypothesis space is a {\it hypothesis}.
Clause $C$ {\it covers} an example $e$ if $I \wedge C \models e$, where $\models$ is the entailment operator, i.e., if $I$ and $C$ are true, then $e$ is true. 
Definition $H$ covers an example $e$ if at least one its clauses covers $e$.
The goal of a learning algorithm is to find the definition in the hypothesis space that covers 
all positive and the fewest negative examples as possible. 

\begin{example}
IMDb contains the tuples {\it movie (10,`Star Wars: Episode IV - 1977', 1977), mov2genres(10, `comedy'), and \\ mov2releasedate(10, `May', 1977)}. 
Therefore, the definition that indicates that high grossing movies are often released in May and their genre is {\it comedy} shown in Section~\ref{section:heterolearn-introduction} covers the positive example {\it highGrossing(`Star Wars: Episode IV - 1977')}.
\end{example}

Most relational learning algorithms follow a covering approach illustrated in Algorithm~\ref{algorithm:castor-covering}  \cite{Mihalkova:ICML:07,progolem,castor:SIGMOD17,Quinlan:FOIL,QuickFOIL}. 
The algorithm constructs one clause at a time using the {\it LearnClause} function. 
If the clause satisfies a criterion, e.g., covers at least a certain fraction of the positive examples and does {\it not} cover more than a certain fraction of negative ones,
the algorithm adds the clause to the learned definition and discards the positive examples covered by the clause. 
It stops when all positive examples are covered by the learned definition. 
\begin{algorithm} 
	\SetKwInOut{Input}{Input}
	\SetKwInOut{Output}{Output}
	\Input{Database instance $I$, examples $E$}
	\Output{Horn definition $H$}
	$H = \{\}$\;
	$U = E^+$\;
	\While{$U$ is not empty }{
		$C =$ LearnClause$(I,U,E^-)$\;
		\If{$C$ satisfies minimum criterion} {
			$H = H \cup C$\;
			$U = U - \{ e \in U | H \wedge I \models e \}$\;
		}
	}
	return $H$ \;
	\caption{Covering approach algorithm.}
	\label{algorithm:castor-covering}
\end{algorithm}
\vspace{-10pt}

\subsection{Matching Dependencies}
\label{sec:MD}
Learning over databases with heterogeneity in representing values may deliver inaccurate answers as the same entities and values may be represented under different names. Thus, one must resolve these representational differences to produce a high-quality database to learn an effective definition.
The database community has proposed declarative matching and resolution rules to express the domain knowledge about matching and resolution  \cite{4497412,DBLP:conf/kr/BahmaniBKL12,Benjelloun:2009:SGA:1541533.1541538,Burdick:2016:DFL:2966276.2894748,Fan2009ReasoningAR,Galhardas:2001:DDC:645927.672042,Hernandez:2013:HHS:2452376.2452440,Hernandez:1995:MPL:223784.223807,Weis:2008:IDD:1454159.1454165}.
{\it Matching dependencies (MD)} are a popular type of such declarative rules, which provide a powerful method of expressing domain knowledge on matching values \cite{Fan2009ReasoningAR,Bertossi2011DataCA,Bahmani2015ERBloxCM,DBLP:series/synthesis/2012Fan,MDedup:Koumarelas}.
Let $\mS$ be the schema of the original database and $R_1$ and $R_2$ two distinct relations in $\mS$.
Attributes $A_1$ and $A_2$ from relations $R_1$ and $R_2$, respectively, are comparable if they share the dame domain.
MD $\sigma$ is a sentence of the form $R_1[A_1] \approx_{dom(A_1)} R_2[B_1],$ $\ldots,$ 
$R_1[A_n] \approx_{dom(A_n)} R_2[B_n]$ 
$\rightarrow$ $R_1[C_1] \rightleftharpoons R_2[D_1],$ $\dots,$ $R_1[C_m] \rightleftharpoons R_2[D_m]$, 
where $A_i$ and $C_j$ are comparable to $B_i$ and $D_j$, respectively, $1 \leq i \leq n$, 
and $1 \leq j \leq m$. Operation $\approx_{d}$ is a similarity operator defined over domain $d$ and
$R_1[C_j] \rightleftharpoons R_2[D_j]$,$1 \leq j \leq m$, indicates that the values of 
$R_1[C_j]$ and $R_2[D_j]$ refer to the same value, i.e., are interchangeable.
Intuitively, the aforementioned MD says that if the values of $R_1[A_i]$ and $R_2[B_i]$ are sufficiently similar, the values of $R_1[C_j]$ and $R_2[D_j]$ are different representations of the same value. 
For example, consider again the database that contains relations from {\it IMDb} and {\it BOM} whose schema fragments are shown in Table~\ref{table:schemas}. According to our discussion in Section~\ref{section:heterolearn-introduction}, one can define the following MD $\sigma_1:$
$\textit{movies[title]} \approx \textit{highBudgetMovies[title]}$ $\rightarrow$ $\textit{movies[title]} \rightleftharpoons \textit{highBudgetMovies[title]}$.
The exact implementation of the similarity operator depends on the underlying domains of attributes. 
Our results are orthogonal to the implementation details of the similarity operator. 
In the rest of the paper, we use $\approx_{d}$ operation only between comparable attributes.
For brevity, we eliminate the domain $d$ from $\approx_{d}$ when it is clear from the context or the 
results hold for any domain $d$.
We also denote $R_1[A_1] \approx R_2[B_1],$ $\ldots,$ $R_1[A_n] \approx R_2[B_n]$ in an MD
as $R_1[A_{1 \ldots n}] \approx $ $R_2[B_{1 \ldots n}]$.
An MD $R_1[A_{1 \ldots n}] \approx $ $R_2[B_{1 \ldots n}]$
$\rightarrow$ $R_1[C_1] \rightleftharpoons R_2[D_1],$ $\dots,$ $R_1[C_m] \rightleftharpoons R_2[D_m]$ is 
equivalent to a set of MDs  $R_1[A_{1 \ldots n}] \approx $ $R_2[B_{1 \ldots n}]$
$\rightarrow$ $R_1[C_1] \rightleftharpoons R_2[D_1]$, 
$R_1[A_{1 \ldots n}] \approx $ $R_2[B_{1 \ldots n}]$ $\rightarrow$ $R_1[C_2] \rightleftharpoons R_2[D_2]$, 
$\dots$, 
$\rightarrow$ $R_1[C_1] \rightleftharpoons R_2[D_1]$ $\rightarrow$ $\dots,$ $R_1[C_m] \rightleftharpoons R_2[D_m]$.
Thus, for the rest of the paper, we assume that each MD is in the form of 
$R_1[A_{1 \ldots n}] \approx $ $R_2[B_{1 \ldots n}]$ $\rightarrow$ $R_1[C] \rightleftharpoons R_2[D]$, 
where $C$ and $D$ are comparable attributes of $R_1$ and $R_2$, respectively.
Given a database with MDs, one must enforce the MDs to generate a high-quality database. 
Let tuples $t_1$ and $t_2$ belong to $R_1$ and $R_2$ in database $I$ of schema $\mS$, respectively, such that $t_1^{I}[A_i] \approx t_2^{I}[B_i]$, $1 \leq i \leq n$, denoted as 
$t_1^{I}[A_{1 \ldots n}] \approx$ $t_2^{I}[B_{1 \ldots n}]$ for brevity.
To enforce the MD $\sigma:$ 
$R_1[A_{1 \ldots n}] \approx $ $R_2[B_{1 \ldots n}] \rightarrow$ 
$R_1[C] \rightleftharpoons R_2[D]$ on $I$, one must make the values
of $t_1^{I}[C]$ and $t_2^{I}[D]$ identical as they actually refer to the same value \cite{Fan2009ReasoningAR,Bertossi2011DataCA}. 
For example, if attributes $C$ and $D$ contain titles of movies, 
one unifies both values {\it Star Wars - 1977} and {\it Star Wars - IV} to 
{\it Star Wars Episode IV - 1977} as it deems this value as the one to which 
$t_1^{I}[C]$ and $t_2^{I}[D]$ refer.
The following definition formalizes the concept of applying an MD to the tuples $t_1$ and $t_2$ on $I$. 
\begin{definition}
\label{def:MDEnforce}
Database $I'$ is the immediate result of enforcing MD $\sigma$ on $t_1$ and $t_2$ in $I$, denoted by $(I, I')_{[t_1,t_2]} \models \sigma$ if  
\begin{enumerate} 
	\item $t_1^{I}[A_{1 \ldots n}] \approx t_2^{I}[B_{1 \ldots n}]$, but $t_1^{I}[C] \neq t_2^{I}[D]$;
	\item $t_1^{I'}[C] = t_2^{I'}[D]$ ; and
	\item $I$ and $I'$ agree on every other tuple and attribute value.
\end{enumerate}
\end{definition}
One may define a unification function over some domains to map the values that refer to the same value to the correct value in the cleaned instance. 
It is, however, usually difficult to define such a function due to the lack of knowledge about the correct value. 
For example, let $C$ and $D$ in Definition~\ref{def:MDEnforce} contain information about 
names of people and $t_1^{I}[C]$ and $t_2^{I}[D]$ have values 
{\it J. Smth} and {\it Jn Sm}, respectively, which according to an MD refer to the same actual name, which is {\it Jon Smith}. 
It is not clear how to compute {\it Jon Smith} using the values of $t_1^{I}[C]$ and $t_2^{I}[D]$. 
We know that the values of $t_1^{I'}[C]$ and $t_2^{I'}[D]$ will be identical after enforcing $\sigma$, 
but we do not usually know their exact values. 
Because we aim at developing learning algorithms that are efficient and effective over databases from various domains, we do {\it not} fix any matching method in this paper. We assume that matching every pair of values $a$ and $b$ in the database creates a fresh value denoted as $v_{a,b}$.

Given the database $I$ with the set of MDs $\Sigma$, $I'$ is {\it stable} if $(I, I')_{[t_1,t_2]} \models \sigma$ for all $\sigma \in \Sigma$ and all tuples $t_1, t_2 \in I'$.
In a stable database instance, all values that represent the same data item according to the database MDs are assigned equal values. Thus, it does not have any heterogeneities. 
Given a database $I$ with set of MDs $\Sigma$, one can produce a stable instance for $I$ by starting from $I$ and
iteratively applying each MD in $\Sigma$ according to Definition~\ref{def:MDEnforce} finitely many times \cite{Fan2009ReasoningAR,Bertossi2011DataCA}.
Let $I, I_1, \ldots, I_k$ denote the sequence of databases produced by applying MDs according to Definition~\ref{def:MDEnforce} starting from $I$ such that $I_k$ is stable. 
We say that $(I, I_k)$ satisfy $\Sigma$ and denote it as $(I, I_k) \models \Sigma$.
A database may have many stable instances depending on the order of MD applications~\cite{Bertossi2011DataCA,Fan2009ReasoningAR}.
\begin{example}
\label{ex:multipleInstance}
	Let {\it (10,`Star Wars: Episode IV - 1977', 1977)} and {\it (40,`Star Wars: Episode III - 2005', 2005)} be tuples in relation {\it movies} and {\it (`Star Wars')} be a tuple in relation {\it highBudgetMovies} whose schemas are shown in Table~\ref{table:schemas}. Consider MD $\sigma_1:$ $\textit{movies[title]} \approx \textit{highBudgetMovies[title]}$ 
	$\rightarrow$ $\textit{movies[title]} \rightleftharpoons$
	$ \textit{highBudgetMovies[title]}$. 
	Let {\it `Star Wars: Episode IV - 1977'} $\approx$ {\it `Star Wars'} and  {\it `Star Wars: Episode III - 2005'} $\approx$ {\it `Star Wars'} be true.
	Since the movies with titles {\it `Star Wars: Episode IV - 1977'} and {\it `Star Wars: Episode III - 2005'} are different movies with distinct titles, one can unify the title in the tuple {\it (`Star Wars')} in {\it highBudgetMovies} with only one of them in each stable instance. 
	Each alternative leads to a distinct instance.
\end{example}
MDs may {\it not} have perfect precision. If two values are declared similar according to an MD, 
it does not mean that they represent the same real-world entities. But, it is more likely for them to represent the same value than the ones that do not match an MD.
Since it may be cumbersome to develop complex MDs that are sufficiently accurate, 
researchers have proposed systems that automatically discover MDs from the database content \cite{MDedup:Koumarelas}.

\vspace{-5pt}
\subsection{Conditional Functional Dependencies}
\label{sec:CFD}
Users usually define integrity constraints (IC) to ensure the quality of the data.
Conditional functional dependencies (CDF) have been useful in defining quality rules for cleaning data
\cite{DBLP:conf/pods/Fan08,10.14778/1952376.1952378,10.14778/1453856.1453900,Cong:Gao:2007,10.14778/1952376.1952378,DBLP:series/synthesis/2012Fan}. 
They extend functional dependencies, which are arguably the most widely used ICs \cite{DBBook}.
Relation $R$ with sets of attributes $X$ and $Y$ satisfies FD $X \rightarrow Y$ 
if every pairs of tuples in $R$ that agree on the values of $X$ will also agree on the values
of $Y$. 
A CFD $\phi$ over $R$ is a form $(X \rightarrow Y, t_{p})$ where $X \rightarrow Y$ is an FD over $R$
and $t_p$ is a tuple pattern over $X \cup Y$. For each attribute $A \in$ $X \cup Y$, 
$t_p[A]$ is either a constant in domain of $A$ or an unnamed variable denoted as `-' that takes values from the domain of $A$. The attributes in $X$ and $Y$ are separated by $\mid\mid$ in $t_p$.
For example, consider relation {\it mov2locale}({\it title}, {\it language}, {\it country}) in BOM.
The CFD $\phi_1$: ({\it title}, {\it language} $\rightarrow$ {\it country}, (-, English $\mid\mid$ -) )
indicates that {\it title} uniquely identifies {\it country} for tuples whose language is English. 
Let $\asymp$ be a predicate over data values and unnamed variable `-', where $a \asymp b$ if either
$a=b$ or $a$ is a value and $b$ is `-'. The predicate $\asymp$ naturally extends to tuples, e.g., 
(`Bait', English, USA) $\asymp$ (`Bait', -, USA). Tuple $t_1$ {\it matches} $t_2$ if $t_1 \asymp$ $t_2$.
Relation $R$ satisfies the CFD $(X \rightarrow Y, t_{p})$ iff for each pair of tuples 
$t_1, t_2$ in the instance if $t_1[X] = t_2[X]$ $\asymp t_p[X]$, then $t_1[Y] = t_2[Y]$ $\asymp t_p[Y]$.
In other words, if $t_1[X]$ and $t_2[X]$ are equal and match pattern $t_p[X]$, $t_1[Y]$ and 
$t_2[Y]$ are equal and match $t_p[Y]$. 
A relation satisfies a set of CFDs $\Phi$, if it satisfies every CFD in $\Phi$.
For each set of CFDs $\Phi$, we can find an equivalent set of CFDs whose members have a single 
attribute on their right-hand side \cite{Cong:Gao:2007,10.14778/1952376.1952378,DBLP:series/synthesis/2012Fan}. 
For the rest of the paper, we assume that each CFD has a single attribute on its right-hand side.

CFDs may be violated in real-world and heterogeneous datasets \cite{10.14778/1952376.1952378,10.14778/1453856.1453900}. For example, 
the pair of tuples $r_1:$(`Bait', English, USA) and $r_2:$(`Bait', English, Ireland) in 
{\it movie2locale} violate $\phi_1$.
One can use attribute value modifications to repair violations of a CFD in a relation and generate a 
repaired relation that satisfy the CFD \cite{10.1145/1066157.1066175,10.1007/3-540-45653-8_39,10.5555/645505.656435,Cong:Gao:2007,10.14778/1952376.1952378,10.1145/1514894.1514901,DBLP:series/synthesis/2012Fan}. 
For instance, one may repair the violation of $\phi_1$ in $r_1$ and $r_2$ by
updating the value of title or language in one of the tuples to value other than Bait or English, respectively.
One may also repair this violation by replacing the countries in these tuples with the same value.
Inserting new tuples do not repair CFD violations and one may simulate tuple deletion using 
value modifications. Moreover, removing tuples leads to unnecessary loss of information for attributes 
that do not participate in the CFD.
Modifying attribute values is also sufficient to resolve CFD violations \cite{Cong:Gao:2007,10.14778/1952376.1952378}.
Thus, given a pair of tuples $t_1$ and $t_2$ in $R$ that violate CFD $(X \rightarrow A, t_{p})$, 
to resolve the violation, one must either modify $t_1[A]$ (resp. $t_2[A]$) such that $t_1[A] = t_2[A]$ and
$t_1[A] \asymp t_p[A]$, update $t_1[X]$ (resp. $t_2[X]$) such that $t_1[X] \nasymp t_p[X]$ 
(resp. $t_2[X] \nasymp t_p[X]$) or $t_1[X] \neq t_2[X]$. 
Let $R$ be a relation that violates CFD $\phi$. 
Each updated instance of $R$ that is generated by applying the aforementioned repair operations 
and does not contain any violation of $\phi$ is a {\it repair} of $R$.
As there are multiple fixes for each violation, there may be many repairs for each relation.

As opposed to FDs, a set of CFDs may be inconsistent, i.e., there is not any non-empty database that satisfies them \cite{4221723,Cong:Gao:2007,10.14778/1952376.1952378,DBLP:series/synthesis/2012Fan}.
For example, the CFDs $(A \rightarrow B, a_1||b_1)$ and $(B \rightarrow A, b_1||a_2)$ over relation $R(A,B)$
cannot be both satisfied by any non-empty instance of $R$. The set of CFDs used in cleaning is consistent
\cite{4221723,Cong:Gao:2007,10.14778/1952376.1952378,DBLP:series/synthesis/2012Fan}. We refer the reader to
\cite{4221723} for algorithms to detect inconsistent CFDs.

\section{Semantic of Learning}
\label{section:approach-explanation}

\subsection{Different Approaches}
\label{sec:ApproachesToLearning}
Let $I$ be an instance of schema $\mS$ with MDs $\Sigma$ that violate some CFDs $\Phi$.
A {\it repair} of $I$ is a stable instance of $I$ that satisfy $\Phi$. The values in $I$ are repaired
to satisfy $\Phi$ using the method explained in Section~\ref{sec:CFD}.
Given $I$ and a set of training examples $E$, we wish to learn a definition for a target relation $T$ 
in terms of the relations in $\mS$. 
Obviously, one may not learn an accurate definition by applying current learning algorithms over 
$I$ as the algorithm may consider different occurrences of the same value to be distinct or learn patterns that 
are induced based on tuples that violate CFDs.
One can learn definitions by generating all possible repairs of $I$,
learning a definition over each repair separately, and computing a union (disjunction) of all learned definitions. 
Since the discrepancies are resolved in repaired instances, this approach may learn accurate definitions. 

However, this method is neither desirable nor feasible for large databases. 
As a large database may have numerous repairs, it takes a great deal of time and storage to compute and materialize all of them.
Moreover, we have to run the learning algorithm once for each repair, which may take an extremely long time.
More importantly, as the learning has been done separately over each repair, it is {\it not} clear whether the final definition is sufficiently effective considering the information of all stable instances. 
For example, let database $I$ have two repairs $I^s_1$ and $I^s_2$ over which the aforementioned approach learns definitions $H_1$ and $H_2$, respectively. $H_1$ and $H_2$ must cover a relatively small number of negative examples over $I^s_1$ and $I^s_2$, respectively. 
However, $H_1$ and $H_2$ may cover a lot of negative examples over $I^s_2$ and $I^s_1$, respectively. 
Thus, the disjunction of $H_1$ and $H_2$ will {\it not be effective} considering the information 
in both $I^s_1$ and $I^s_2$. 
Hence, it is not clear whether the disjunction of $H_1$ and $H_2$ is the definition that covers all positive and the least negative examples over $I^s_2$ and $I^s_1$. 
Also, it is not clear how to encode usably the final result as we may end up with numerous definitions.

Another approach is to consider only the information shared among all repairs for learning. 
The resulting definition will cover all positive and the least negative examples considering the information common among all repaired instances. 
This idea has been used in the context of query answering over inconsistent data, i.e., consistent query answering \cite{Arenas:PODS:99:Consistent,Bertossi2011DataCA}. 
However, this approach may lead to ignoring many positive and negative examples as their connections to other relations in the database may {\it not} be present in all stable instances.
For example, consider the tuples in relations {\it movies} and {\it highBudgetMovies} in Example~\ref{ex:multipleInstance}.
The training example {\it (`Star Wars')} has different values in different stable instances of the database, therefore, it will be ignored.
It will also be connected to two distinct movies with vastly different properties in each instance.
Similarly, repairing the instance to satisfy the violated CFDs may further reduce the amount of training 
examples shared among all repairs.
The training examples are usually costly to obtain and the lack of enough training examples may results in inaccurate learned definitions.
Because in a sufficiently heterogeneous database, most positive and negative examples may not be common 
among all repairs, the learning algorithm may learn an inaccurate or simply an empty definition.

Thus, we hit a middle-ground. We follow the approach of learning directly over the original database. 
But, we also give the language of definitions and semantic of learning enough flexibility to take advantage of as much (training) information as possible. 
Each definition will be a compact representation of a set of definitions, each of which is sufficiently accurate over some repairs. 
If one increases the expressivity of the language, learning and checking coverage for each clause may become inefficient \cite{Eiter:1997:DD:261124.261126}. 
We ensure that the added capability to the language of definitions is minimal so learning remains efficient. 

\subsection{Heterogeneity in Definitions}
\label{sec:HetroDef}
We represent the heterogeneity of the underlying data in the language of the learned definitions. 
Each new definition encapsulates the definitions learned over the repairs of the underlying database.
Thus, we add the similarity operation, $x \approx y$, to the language of Horn definitions.
We also add a set of new (built-in) relation symbols $V_{c}$ with arity two called {\it repair relations}
to the set of relation symbols used by the Datalog definitions over schema $\mS$. 
A literal with a repair relation symbol is a {\it repair literal}. 
Each repair literal $V_{c}(x, v_x)$ in a definition $H$ represents replacing 
the variable (or constant) $x$ in (other) existing literals in $H$ with variable $v_x$ if 
condition $c$ holds. Condition $c$ is a conjunction of 
$=$, $\neq$, and $\approx$ relations over the variables and constants in the clause. 
Each repair literal reflects a repair operation explained in Sections~\ref{sec:MD} and
\ref{sec:CFD} for an MD or violated CFD over the underlying database.
The condition $c$ is computed according to the corresponding MD or CFD.
Finally, we add a set of literals with $=$, $\neq$, and $\approx$ relations called {\it restriction literals}
to establish the relationship between the replacement variables, e.g, $v_x$, according to the 
corresponding MDs and CFDs.
Consider again the database created by integrating IMDb and BOM datasets, whose schema fragments are in Table~\ref{table:schemas},
with MD $\sigma_1:$ $\textit{movies[title]} \approx \textit{highBudgetMovies[title]}$ $\rightarrow$ $\textit{movies[title]} \rightleftharpoons \textit{highBudgetMovies[title]}$. 
We may learn the following definition for the target relation {\it highGrossing}.
\begin{align*}
		\mathit{highGrossing}(x) \leftarrow & \mathit{movies}(y, t, z), \mathit{mov2genres}(y,\mathit{`comedy\text{'}}),  \\
		& \mathit{highBudgetMovies}(x), x\approx t, \mathit{V_{x\approx t}}(x,v_x), \\
		& \mathit{V_{x\approx t}}(t,v_t), v_x = v_t.
\end{align*}
\noindent
The repair literals $\mathit{V_{x\approx t}}(x,v_x)$ and $\mathit{V_{x\approx t}}(t,v_t)$ represent
the repairs applied to $x$ and $t$ to unify their values to a new one according to $\sigma_1$.
We add equality literal $v_x=v_t$ to restrict the replacements according to the corresponding MD.

We also use repair literals to fix a violation of a CFD in a clause. These repair literals reflect the operations
explained in Section~\ref{sec:CFD} to fix the violation of a CFD in a relation. 
The resulting clause represents possible repairs for a violation of a CFD in the clause. 
A variable may appear in multiple literals in the body of a clause and some repairs may modify only some of the occurrences of the variable, e.g., the example on BOM database in Section~\ref{sec:CFD}.
Thus, before adding repair literals for both MDs and CFDs, 
we replace each occurrence of a variable with a fresh one and add equality literals, i.e., {\it induced equality literals}, to maintain the connection between their replacements. 
Similarly, we replace each occurrence of the constant with a fresh variable and 
use equality literals to set the value the variable equal to the constant in the clause.
\vspace{-5pt} 
\begin{example}
\label{ex:CFDRepair}
Consider the following clause, that may be a part of a learned clause over the integrated 
IMDb and BOM database for {\it highGrossing}.
\begin{align*}
		\mathit{highGrossing}(x) \leftarrow & \mathit{mov2locale}(x, English, z),\\
		&\mathit{mov2locale}(x, English, t).
\end{align*}
\noindent
This clause reflects a violation of CFD $\phi_1$ from Section~\ref{sec:CFD} in the underlying database as
it indicates that English movies with the same title are produced in different countries.  
We first replace each occurrence of repeated variable $x$ 
with a new variable and then add the repair literals. Due to the limited space, we do not show the repair literals and their conditions for modifying the values of constant 'English'. Let condition $c$ be $x_1=x_2 \land z\neq t$.
\begin{align*}
		\mathit{highGrossing}&(x_1) \leftarrow  \mathit{mov2locale}(x_1, English, z),\\
		& \mathit{mov2locale}(x_2, English, t), x_1= x_2, \mathit{V_c}(x_1,v_{x_1}),\\
		& \mathit{V_{c}}(x_2,v_{x_2}), v_{x_1} \neq x_2, v_{x_2} \neq x_1, \mathit{V_c}(z,t),\\
		& \mathit{V_c}(t,z), \mathit{V_c}(z,v_{z}), \mathit{V_c}(t,v_{t}), v_{z} = v_{t}.
\end{align*}
\end{example}
\noindent

We call a clause (definition) {\it repaired} if it does {\it not} have any repair literal.
Each clause with repair literals represents a set of repaired clauses. 
We convert a clause with repair literals to a set of repaired clauses by iteratively applying repair literals to
and eliminating them from the clause. 
To apply a repair literal $V_c(x,v_x)$ to a clause, we first evaluate $c$ considering the (restriction) literals 
in the clause. If $c$ holds, we replace all occurrences of $x$ with $v_x$ in all literals and the conditions of the other repair literals in the clause and remove $V_c(x,v_x)$. Otherwise, we only eliminate $V_c(x,v_x)$ from 
the clause. We progressively apply all repair literals until no repair literal is left.
Finally, we remove all restriction and induced equality literals that contain at least one variable that does not appear in any literal with a schema relation symbol.
The resulting set is called the {\it repaired clauses} of the input clause. 
\vspace{-5pt}
\begin{example}
	\label{ex:MDapplication}
	Consider the following clause over the movie database of IMDb and BOM.
	\begin{align*}
	\mathit{highGrossing}(x) \leftarrow&  \mathit{movies}(y, t, z), \mathit{mov2genres}(y,\mathit{`comedy\text{'}}), \\
	& \mathit{highBudgetMovies}(x), x\approx t, \mathit{V_{x\approx t}}(x,v_x), \\
	& \mathit{V_{x\approx t}}(t,v_t), v_x = v_t.
	\end{align*}
	The application of repair literals $\mathit{V_{x\approx t}}(x,v_x)$ and $\mathit{V_{x\approx t}}(t,v_t)$  results in the following clause.
	\begin{align*}
	\mathit{highGrossing}(v_{x}) \leftarrow&  \mathit{movies}(y, v_{t}, z), \mathit{mov2genres}(y,\mathit{`comedy\text{'}}),\\
	& \mathit{highBudgetMovies}(v_{x}), v_x = v_t.
	\end{align*}
\end{example}
\noindent

Similar to the repair of a database based on MDs and CFDs, 
the application of a set of repair literals to a clause may create multiple repaired clauses depending on the order
by which the repair literals are applied.
\begin{example}
\label{example:multipleStableClauses}
Consider a target relation $T(A)$, an input database with schema $\{R(B)$, $S(C)\}$, and MDs
$\phi_1:$ $T[A] \approx R[B] \rightarrow T[A] \rightleftharpoons R[B]$ and $\phi_2:$ $T[A] \approx S[C] \rightarrow T[A] \rightleftharpoons S[C]$.
The definition $H:$ $T(x) \leftarrow R(y), x \approx y, V_{x \approx y}(x, v_x),$ 
$V_{x \approx y}(y, v_y), v_x = v_y, S(z), x \approx z,$ $V_{x \approx z}(x, u_x),$ 
$V_{x \approx z}(z, v_z), u_x = v_z.$ 
over this schema has two repaired definitions: 
$H'_1:$ $T(v_x) \leftarrow$ $R(v_y), v_x = v_y, S(z).$ and 
$H'_2:$ $T(u_x) \leftarrow$ $R(y), S(v_z), u_x = v_z.$
As another example, the application of each repair literal 
in the clause of Example~\ref{ex:CFDRepair} results in 
a distinct repaired clause. For instance, applying $\mathit{V_c}(x_1,v_{x_1})$ replaces $x_1$ with
$v_{x_1}$ in all literals and conditions of the repair literals and results in the following.
\begin{align*}
		\mathit{highGrossing}&(v_{x_1}) \leftarrow  \mathit{mov2locale}(v_{x_1}, English, z),\\
		& \mathit{mov2locale}(x_2, English, t),\mathit{V_{c}}(x_2,v_{x_2}), v_{x_1} \neq x_2. 
\end{align*}
\end{example}
\noindent 
As Example~\ref{example:multipleStableClauses} illustrates, repair literals provide a compact representation of multiple learned clauses where each may explain the patterns in the training data in some repair of the input database.
Given an input definition $H$, the {\it repaired definitions} of $H$ are a set of definitions where each one contains exactly one repaired clause per each clause in $H$.


\subsection{Coverage Over Heterogeneous Data}
A learning algorithm evaluates the score of a definition according to the number of its covered positive and negative examples. 
One way to measure the score of a definition is to compute the difference of the number of positive and negative examples covered by the definition \cite{DeRaedt:2010:LRL:1952055,Quinlan:FOIL,QuickFOIL}. 
Each definition may have multiple repaired definitions each of which may cover a different number of positive and negative examples on the repairs of the underlying database. 
Thus, it is not clear how to compute the score of a definition.

One approach is to consider that a definition covers a positive example if at least one of 
its repaired definitions covers it in some repaired instances. 
Given all other conditions are the same, this approach may lead to learning a definition with numerous 
repaired definitions where each may not cover sufficiently many positive examples. 
Hence, it is not clear whether each repaired definition is accurate.
A more restrictive approach is to consider that a definition covers a positive example if all 
its repaired definitions cover it. 
This method will deliver a definition whose repaired definitions have high positive coverage over repaired instances. 
There are similar alternatives for defining coverage of negative examples. 
One may consider that a definition covers a negative example if all of its repaired definitions cover it. 
Thus, if at least one repaired definition does not cover the negative example, 
the definition will not cover it. 
This approach may lead to learning numerous repaired definitions, which cover many negative examples. 
On the other hand, a restrictive approach may define a negative example covered by a definition if at least one of its repaired definitions covers it. 
In this case, generally speaking, each learned repaired definition will {\it not} cover too many negative examples.
We follow a more restrictive approach. 
\begin{definition}
\label{def:coverExample}
A definition $H$ covers a positive example $e$ w.r.t. to database $I$ iff every 
repaired definition of $H$ covers $e$ in some repairs of $I$.
\end{definition}
\noindent
\begin{example}
\label{example:CoverageHetero}
Consider again the schema, MDs, and definition $H$ in 
Examples~\ref{example:multipleStableClauses} and the database of this schema with training example $T(a)$ and tuples $\{ R(b), S(c)\}$. 
Assume that $a \approx b$ and $a \approx c$ are true. 
The database has two stable instances $I'_1:$ $\{T(v_{a,b}), R(v_{a,b})$, $S(c)\}$ and $I'_2:$ $\{ T(v_{a,c}), R(b)$, $S(v_{a,c})\}$. 
Definition $H$ covers the single training example in the original database according to Definition~\ref{def:coverExample} as its repaired definitions $H'_1$ and $H'_2$ cover the training example in repaired instances $I'_1$ and $I'_2$, respectively. 
\end{example}
\noindent
Definition~\ref{def:coverExample} provides a more flexible semantic than considering only the common information between all repaired instances as described in Section~\ref{sec:ApproachesToLearning}. 
The latter semantic considers that the definition $H$ covers a positive example if it covers the example in all repaired instances of a database. As explained in Section~\ref{sec:ApproachesToLearning}, this approach may lead to ignoring many if not all examples. 

\begin{definition}
\label{def:coverExampleNegative}
A definition $H$ covers a negative example $e$ with regard to database $I$ if at least one of the repaired definitions of $H$ covers $e$ in some repairs of $I$.
\end{definition}

		\section{DLearn}
\label{section:HDLEARN}
In this section, we propose a learning algorithm called {\it DLearn} for learning over heterogeneous data efficiently.  It follows the approach used in the bottom-up relational learning algorithms \cite{golem,progolem,Mihalkova:ICML:07,castor:SIGMOD17}.
In this approach, the {\it LearnClause} function in Algorithm~\ref{algorithm:castor-covering} has two steps. It first builds the most specific clause in the hypothesis space that covers a given positive example, called a {\it bottom-clause}. 
Then, it generalizes the bottom-clause to cover as most positive and as fewest negative examples as possible. 
DLearn extends these algorithms by integrating the input MDs and CFDs into the 
learning process to learn over heterogeneous data. 


\subsection{Bottom-clause Construction}
\label{section:bottom-clause-construction}
A {\it bottom-clause} $C_e$ associated with an example $e$ is the most specific clause in the hypothesis space that covers $e$ relative to the underlying database $I$. 
Let $I$ be the input database of schema $\mS$ and the set of MDs $\Sigma$ and CFDs $\Phi$.
The bottom-clause construction algorithm consists of two phases.
First, it finds all the information in $I$ relevant to $e$.
The information relevant to example $e$ is the set of tuples $I_{e} \subseteq I$ that are connected to $e$.
A tuple $t$ is connected to $e$ if we can reach $t$ using a sequence of exact or approximate 
(similarity) matching operations, starting from $e$.
Given the information relevant to $e$, DLearn creates the bottom-clause $C_e$.

\begin{table}
	\centering
	\caption{ Example movie database. }
	\vspace{-10pt}
	\begin{tabular} { l l }
		\hline
		{\small movies(m1,Superbad (2007),2007)}&  {\small mov2genres(m1,comedy)} \\
		{\small movies(m2,Zoolander (2001),2001)} & {\small mov2genres(m2,comedy)} \\
		{\small movies(m3,Orphanage (2007),2007)} & {\small mov2genres(m3,drama)} \\
		{\small mov2countries(m1,c1)} &  {\small countries(c1,USA)}\\
		{\small mov2countries(m2,c1)} & {\small countries(c2,Spain)} \\
		{\small mov2countries(m3,c2)} &  {\small englishMovies(m1)} \\
		{\small mov2releasedate(m1,August,2007)} & {\small englishMovies(m2)} \\ 
		{\small mov2releasedate(m2,September,2001)} & s{\small englishMovies(m3)} \\
		\hline
	\end{tabular}
	\label{table:example}
	\vspace{-15pt}
\end{table}

\begin{example}
	\label{example-relevant-tuples}
	Given example {\it highGrossing(Superbad)}, database in Table~\ref{table:example}, and MD
	$\sigma_2:$ $\mathit{highGrossing[title]} \approx \mathit{movies[title]}$ $\rightarrow 
	\mathit{highGrossing[title]} \rightleftharpoons \mathit{movies [title]}$,
	DLearn finds the relevant tuples {\it movies(m1, Superbad (2007), 2007), 
		mov2genres(m1, comedy), 
		mov2countries(m1, c1), englishMovies(m1),\\ mov2releasedate(m1, August, 2007)}, and {\it countries(c1, USA)}.
	As the movie title in the training example, e.g., {\it Superbad}, does not exactly match with the movie title in the movies relation, e.g., {\it Superbad (2007)}, the tuple {\it movies(m1, Superbad (2007), 2007)} is obtained through an approximate match and similarity search according to $\sigma_2$. 
	We get others via exact matches.
\end{example}

To find the information relevant to $e$, DLearn uses Algorithm~\ref{algorithm:bottom-clause-original-mds}.
It maintains a set $M$ that contains all seen constants.
Let $e=T(a_1, \ldots, a_n)$ be a training example.
First, DLearn adds $a_1, \ldots, a_n$ to $M$.
These constants are values that appear in tuples in $I$.
Then, DLearn searches all tuples in $I$ that contain at least one constant in $M$ and adds them to $I_{e}$.
For exact search, DLearn uses simple SQL selection queries over the underlying relational database.
For similarity search, DLearn uses MDs in $\Sigma$.
If $M$ contains constants in some relation $R_i$ and 
given an MD $\sigma' \in \Sigma$, 
$\sigma':$ $R_1[A_{1 \ldots n}] \approx $ $R_2[B_{1 \ldots n}]$  $\rightarrow R_1[C] \rightleftharpoons R_2[D]$
DLearn performs a similarity search over $R_2[B_j]$, $1 \leq j \leq n$ 
to find relevant tuples in $R_2$, denoted by $\psi_{B_i \approx M}(R_2)$.
We store these pairs of tuples that satisfy the similarity match in $I_{e}$ in a table in main memory.
We will discuss the details of the implementation of DLearn over relational database systems in Section~\ref{section:heterolearn-implementation-details}. 
For each new tuple in $I_{e}$, the algorithm extracts new constants and adds them to $M$.
It repeats this process for a fixed number of iterations $d$.

\begin{algorithm}
	\SetKwInOut{Input}{Input}
	\SetKwInOut{Output}{Output}
	\Input{example $e$, \# of iterations $d$} 
	\Output{bottom-clause $C_e$}
	
	$I_{e} = \{\}$\;
	$M = \{\}$ // $M$ stores known constants\;
	
	add constants in $e$ to $M$\;
	
	\For{$i=1$ to $d$}{
		\ForEach{relation $R \in I$}{
			\ForEach{attribute $A$ in $R$}{
				// select tuples with constants in M\;
				$I_R = \sigma_{A \in M}(R)$\;
				
				\If{$\exists$ MD $\sigma' \in \Sigma,$ $\sigma':$ $R_1[A_{1 \ldots n}] \approx $ $R_2[B_{1 \ldots n}]$ $\rightarrow R_1[C] \rightleftharpoons R_2[D]$}{
					$I_R = I_R \cup \psi_{B_j \approx M}(R)$, $1 \leq j \leq n$\;
				}
				
				\ForEach{tuple $t \in I_R$}{
					add $t$ to $I_{e}$ and constants in $t$ to $M$\;
				}
			}
		}
	}
	$C_e=$ create clause from $e$ and $I_{e}$\;
	return $C_e$\;
	
	\caption{DLearn bottom-clause construction algorithm.}
	\label{algorithm:bottom-clause-original-mds}
\end{algorithm}
To create the bottom-clause $C_e$ from $I_{e}$, DLearn first maps each constant in $M$ to a new variable.
It creates the head of the clause by creating a literal for $e$ and replacing the constants in $e$ with their assigned variables.
Then, for each tuple $t \in I_{e}$, DLearn creates a literal and adds it to the body of the clause, replacing each constant in $t$ with its assigned variable.
If there is a variable that appears in more than a single literal, we add the equality literals according to the
method explained in Section~\ref{sec:HetroDef}.
If $t$ satisfies a similarity match according to the table of similarity matches with tuple $t'$, 
we add a similarity literal $s$ per each value match in $t$ and $t'$ to the clause.
Let $\sigma$ be the corresponding MD of this similarity match.
We will also add repair literals $V_s(x,v_x)$ and $V_s(y,v_y)$ and restriction equality literal $v_x = v_y$
to the clause according to $sigma$.
\vspace{-3pt}
\begin{example}
	\label{example-bottomclause}
	Given the relevant tuples found in Example~\ref{example-relevant-tuples}, DLearn creates the following bottom-clause:
	\begin{align*}
	\mathit{high}&\mathit{Grossing}(x) \leftarrow \mathit{movies}(y, t, z), x\approx t,\mathit{V_{x\approx t}}(x,v_x), \\
		&  \mathit{V_{x\approx t}}(t,v_t), v_x = v_t, \mathit{mov2genres}(y,\mathit{`comedy\text{'}}),\\ 
	& \mathit{mov2countries}(y,u),\mathit{countries}(u,\mathit{`USA\text{'}}),\\ 
	& \mathit{englishMovies}(y), \mathit{mov2releasedate}(y,\mathit{`August\text{'}}, w).
	\end{align*}
\end{example}
\noindent
Then, we scan $C_e$ to find violations of each CFD in $\Phi$ and add their corresponding repair literals.
Since each CFD is defined over a single table, we first group literals in $C_e$ based on their relation symbols.
For each group with the relation symbol $R$ and CFD $\phi$ on $R$, our algorithm scans the literals in the group,
finds every pair of literals that violate $\phi$, and adds the repair and restriction literals to the group.
We add the repair and restriction literals corresponding to the repair operations explained in Section~\ref{sec:CFD} to the group and consequently $C_e$ as illustrated in Example~\ref{ex:CFDRepair}. 
The added repair literals will not induce any new violation of $\phi$ in the clause \cite{Cong:Gao:2007,10.14778/1952376.1952378,DBLP:series/synthesis/2012Fan}.
However, repairing a violation of $\phi$ may induce violations for anther CFD $\phi'$ over $R$ \cite{DBLP:series/synthesis/2012Fan}. 
For example, consider CFD $\phi_3:$ $(A \rightarrow B, -\mid\mid-)$ and $\phi_4:$ $(B \rightarrow C, -\mid\mid-)$ 
on relation $R(A,B,C)$. Given literals $l_1: R(x_1,y_1,z_1)$ and $l_2: R(x_1,y_1,z_2)$ that violate 
$\phi_4$, our method adds repair literals that replaces $y_1$ in $l_1$ with a fresh variable. 
This repair literal produces a repaired clause that violates $\phi_3$.
Thus, the algorithm repeatedly scans the clause and adds repair and restriction literals to it for all
CFDs until there is a repair for every violation of CFDs both the ones in the original clause and the ones
induced by the repair literals added in the preceding iterations.
The repaired literals for the violations induced by other repair literals will use the replacement variables from 
the violating repair literals as their arguments and conditions.

It may take a long time to generate the clause that contains all repair literals for all original and 
induced violations of every CFD in a large input bottom-clause. Hence, we reduce the number of repair literals 
per CFD violation by adding only the repair literals for the variables of the right-hand side attribute of the CFD
that use current variables in the violation. For instance, in Example~\ref{ex:CFDRepair}, 
the algorithm does {\it not} introduce literals $\mathit{V_c}(z,v_{z})$, $\mathit{V_c}(t,v_{t})$, and
$v_{z} = v_{t}$ and only uses literals $\mathit{V_c}(z,t)$ and $\mathit{V_c}(t,z)$ to repair the clause
in Example~\ref{ex:CFDRepair}. The repair literals for the 
variables corresponding to the left-hand side of the CFD will be used as explained before.
This approach follows the popular minimal repair semantic for repairing CFDs \cite{10.1145/1066157.1066175,10.1007/3-540-45653-8_39,10.5555/645505.656435,Cong:Gao:2007,10.14778/1952376.1952378,10.1145/1514894.1514901,DBLP:series/synthesis/2012Fan} as it repairs the violation 
by modifying fewer variable than the repair literals that introduce fresh variables to the both literals of the
violation, e.g., one versus two modifications induced by $\mathit{V_c}(z,v_{z})$, $\mathit{V_c}(t,v_{t})$ 
in the repair of the clause in Example~\ref{ex:CFDRepair}.
Since each CFD is defined over a single relation, the aforementioned steps are applied separately to literals of each relation, which are usually a considerably smaller set than the set of all literals in the bottom-clause.
Moreover, the bottom-clause is significantly smaller than the size of the whole database.
Thus, the bottom-clause construction algorithm takes significantly less time than producing the repairs of the underlying database.

Current bottom-clause constructions methods do not induce inequality $neq$ literal between distinct constants
in the database and their corresponding variables and represent their relationship by replacing them with distinct 
variables. If the inequality literal is used, the eventual generalization of the 
bottom-clause may be too strict and lead to a learned clause that does not cover sufficiently many positive examples \cite{progolem,progol,DeRaedt:2010:LRL:1952055,castor:SIGMOD17}. 
For example, let $T(x):-$ $R(x,y), S(x,z),$ $y \neq z.$ be a bottom-clause. This clause will not cover
positive examples such as $T(a)$ for which we have $T(a):-$ $R(a,b), S(a,b)$. 
However, the bottom-clause $T(x):-$ $R(x,y), S(x,z)$ has more generalization power and may cover both
positive examples such as $T(a)$ and $T(c)$ such that $T(c):-$ $R(c,b), S(c,d)$. 
As the goal of our algorithm is to simulate relational learning over repaired instances of the original database, 
we follow the same approach and remove the inequality literals between variables. As our repair operations
ensure that the arguments of inequality literals are distinct variables, our method exactly emulates bottom-clause construction in relational learning. 
The inequalities remain in the condition $c$ of each repair literal $V_c$ and will return true if 
the variables are distinct and there is no equality literal between them in the body of the clause
and false otherwise. They are not used in learning and are used to apply repair literals
on the final clause.

\vspace{-4pt}
\begin{proposition}
\label{prop:BC1}
The bottom-clause construction algorithm for positive example $e$ and database $I$ with 
MDs $\Sigma$ and CFD $\Phi$ terminates. Also, the bottom-clause $C_e$ created from $I_e$ 
using the algorithm covers $e$.
\end{proposition}

\subsection{Generalization}
\label{sec:generalization}
After creating the bottom-clause $C_e$ for example $e$, DLearn generalizes $C_e$ to produce a 
clause that is {\it more general} than $C_e$. 
Clause $C$ is more general than clause $D$ if and only if $C$ covers at least 
all positive examples covered by $D$. A more general clause than $C_e$ may cover
more positive examples than $C_e$.
DLearn iteratively applies the generalization to find a clause that covers
the most positive and fewest negative examples as possible.
It extends the algorithm in ProGolem \cite{progolem} to produce generalizations of $C_e$ 
in each step efficiently. This algorithm is based on the concept of $\theta$-subsumption, 
which is widely used in relational learning~\cite{DeRaedt:2010:LRL:1952055,progol,progolem}. 
We first review the concept of $\theta$-subsumption for repaired clauses \cite{DeRaedt:2010:LRL:1952055,progolem}, then, we explain how to extend this concept and its generalization methods for non-stable clauses.

Repaired clause $C$ {\it $\theta$-subsumes} repaired clause $D$, denoted by $C \subseteq_{\theta} D$, 
iff there is some substitution $\theta$ such that $C \theta \subseteq D$ \cite{DeRaedt:2010:LRL:1952055,AliceBook},
i.e., the result of applying substitution $\theta$ to literals in $C$ creates a set of literals that 
is a subset of or equal to the set of literals in $D$.
For example, clause $C_1:$ $\mathit{highGrossing}(x) \leftarrow  \mathit{movies}(x, y, z)$
$\theta$-subsumes $C_2:$
$\mathit{highGrossing}(a)$ $\leftarrow$ $\mathit{movies}(a,$ $b,$ $c),$ $\mathit{mov2genres}(b,\mathit{`comedy\text{'}})$
as for substitution $\theta$ $ =\{ x/a,$ $y/b,$ $z/c\}$, we have $C_1 \theta \subseteq C_2$.
We call each literal $L_D$ in $D$ where there is a literal $L_C$ in $C$ such that 
$L_C\theta = L_D$ a {\it mapped} literal under $\theta$.
For Horn definitions, 
we have $C$ $\theta$-subsumes $D$ iff $C \models G$, i.e., $C$ logically entails $D$ \cite{AliceBook,DeRaedt:2010:LRL:1952055}.
Thus, $\theta$-subsumption is sound for generalization. 
If clauses $C$ and $D$ contain equality and similarity literals, the subsumption checking requires additional
testings, which can be done efficiently \cite{AliceBook,DeRaedt:2010:LRL:1952055,10.1007/978-3-540-24741-8_27}.
Roughly speaking, current learning algorithms generalize a clause $D$ efficiently 
by eliminating some of its literals which produces a clause that $\theta$-subsumes $D$.
We define $\theta$-subsumption for clauses with repair literals using its definition for the repaired ones.
Given a clause $D$, a repair literal $V_c(x,v_x)$ in $D$ is {\it connected to} a non-repair literal $L$ in $D$
iff $x$ or $v_x$ appear in $L$ or in the arguments of a repair literal connected to $L$.
\begin{definition}
\label{def:subsumption}
Let $V(C)$ denote the set of all repair literals in 
$C$ $\theta$-subsumes $D$, denoted by $C \subseteq_{\theta} D$, iff 
\begin{itemize}
\item there is some substitution $\theta$ such that $C \theta \subseteq D$ where repair literals are 
treated as normal ones and 
\item every repair literal connected to a mapped literal in $D$ is also a mapped literal under $\theta$.
\end{itemize}
\end{definition}
\noindent
Definition~\ref{def:subsumption} ensures that each repair literal that modifies a mapped one
in $D$ has a corresponding repair literal in $C$. Intuitively, this guarantees that there is subsumption mapping
between corresponding repaired versions of $C$ and $D$.
The next step is to examine whether $\theta$-subsumption provides a sound bases for
generalization of clauses with repair literals. We first define logical entailment following
the semantics of Definition~\ref{def:coverExample}. 
\begin{definition}
\label{def:entialment}
We have $C \models D$ if and only if there is an onto relation $f$ from the set of repairs of
$C$ to the one of $D$ such that for each repaired clause of $C$, $C_r$, and each $D_r \in f(C_r)$, 
we have $C_r \models D_r$. 
\end{definition}

According to Definitions~\ref{def:entialment}, if one wants to follow the generalization method used in the current learning algorithm to check whether $C$ generalizes $D$, one has enumerate and check $\theta$-subsumption of almost every pair of repaired clauses of $C$ and $D$ in the worst case. 
Since both clauses normally contain many literals and $\theta$-subsumption is NP-hard \cite{AliceBook}, 
this method is not efficient. 
The problem is more complex if one wants to generalize a given clause $D$. It may have to
generate all repaired clauses of $D$ and generalize each of them separately. 
It is not clear how to unify and represent all produced repaired clauses in a 
single non-repaired one. It quickly explodes the hypothesis
space if we cannot represent them in a single clause as the algorithm may have to keep track and 
generalize of almost as many clauses as repairs of the underlying database.
Also, because the learning algorithm performs numerous generalizations and coverage tests, 
learning a definition may take an extremely long time. 
The following theorem establishes that $\theta$-subsumption is sound for generalization of clauses with repair literals.
\begin{theorem}
\label{prop:coveragePositive}
Given clauses $C$ and $D$, if $C$ $\theta$-subsumes $D$, we have $C \models D$. 
\end{theorem}
\vspace{-4pt}
\noindent
To generalize $C_e$, DLearn randomly picks a subset $E^{+s} \subseteq E^+$ of positive examples.
For each example $e'$ in $E^{+s}$, DLearn generalizes $C_e$ to produce a candidate clause $C'$, 
which is more general than $C_e$ and covers $e'$.
Given clause $C_e$ and positive example $e' \in E^{+s}$, DLearn produces a clause that $\theta$-subsumes
$C_e$ and covers $e'$ by removing the {\it blocking literals}.
It first creates a total order between the relation symbols and the symbols of repair literals in 
the schema of the underlying database, e.g., using a lexicographical order and adding the condition
and argument variables to the symbol of the repair literals.
Thus, it establishes an order in each clause in the hypothesis space. 
Let $C_e = T \leftarrow L_1, \cdots, L_n$ be the bottom-clause.
The literal with relation symbol 
$L_i$ is a {\it blocking literal} if and only if $i$ is the least value such that 
for all substitutions $\theta$ where $e' = T\theta$, 
$(T \leftarrow L_1, \cdots, L_i)\theta$ does not cover $e'$~\cite{progolem}.
\vspace{-3pt}
\begin{example}
	Consider the bottom-clause $C_e$ in Example~\ref{example-bottomclause} and positive example {\it $e'=$ highGrossing(`Zoolander\text{'})}. To generalize $C_e$ to cover $e'$, 
	DLearn drops the literal \\
	{\it mov2releasedates}$(y, \mathit{\it `August\text{'}}, u)$ because the movie {\it Zoolander} was not released in August.
\end{example}
\vspace{-4pt}
\noindent
DLearn removes all blocking literals in $C_e$ to produce the generalized clause $C'$.
DLearn also ensures that all literals in the resulting clause are head-connected.
For example, if a non-repair literal $L$ is dropped so as the repair literals whose only 
connection to the head literal is through $L$.
Since $C'$ is generated by dropping literals, it $\theta$-subsumes $C_e$. 
It also covers $e'$ by construction. DLearn generates one clause per example in $E^{+s}$.
From the set of generalized clauses, DLearn selects the highest scoring candidate clause.
The score of a clause is the number of positive minus the number of negative examples covered by the clause.
DLearn then repeats this with the selected clause until its score is not improved.

During each generalization step, the algorithm should ensure that the generalization is minimal with respect to 
$\theta$-subsumption, i.e., 
there is {\it not} any other clause $G$ such that $G$ $\theta$-subsumes $C_e$ and 
$C'$ $\theta$-subsumes $G$ \cite{progolem}. Otherwise, the algorithm may miss some effective clauses 
and produce a clause that is overly general and may cover too many negative examples.
The following proposition states that DLearn produces a minimal generalization in each step.
\vspace{-4pt}
\begin{proposition}
\label{prop:minimal}
Let $C$ be a head-connected and ordered clause generated from a bottom-clause using DLearn
generalization algorithm. Let clause $D$ be the generalization of $C$ produced 
in a single generalization step by the algorithm. Given the clause $F$ that $\theta$-subsumes $C$,
if $D$ $\theta$-subsumes $F$, then $D$ and $F$ are equivalent. 
\end{proposition}
\subsection{Efficient Coverage Testing}
\label{sec:coverage}
DLearn checks whether a candidate clause covers training examples in order to find blocking literals in a clause.
It also computes the score of a clause by computing the number of training examples covered by the clause.
Coverage tests dominate the time for learning~\cite{DeRaedt:2010:LRL:1952055}.
One approach to perform a coverage test is to transform the clause into a SQL query and evaluate it over the input database to determine the training examples covered by the clause. However, since bottom-clauses over large databases normally have many literals, e.g., hundreds of them, the SQL query will involve long joins, making the evaluation extremely slow. Furthermore, it is challenging to evaluate clauses using this approach over heterogeneous data \cite{Bertossi2011DataCA}. 
It is also not clear how to evaluate clauses with repair literals.

We use the concept of $\theta$-subsumption for clauses with repair literals and the 
result of Theorem~\ref{prop:coveragePositive} to compute coverage efficiently.
To evaluate whether $C$ covers a positive example $e$ over database $I$, 
we first build a bottom-clause $G_e$ for $e$ in $I$ called a {\it ground bottom-clause}.  
Then, we check whether $C \wedge I \models e$ using $\theta$-subsumption.
We first check whether $C \subseteq_{\theta} G_e$. Based on Theorem~\ref{prop:coveragePositive}, if we find a substitution $\theta$ for $C$ such that $C \theta \subseteq G_{e}$, and $C$ logically entails $G_e$, thus, 
$C$ covers $e$. 
However, if we cannot find such a substitution, it is not clear whether $C$ logically entails $G_e$ 
as Theorem~\ref{prop:coveragePositive} does not provide the necessity of  
$\theta$-subsumption for logical entailment. Fortunately, this is true if we have only repair literals 
for MDs in $C$ and $G_e$. 
\vspace{-5pt}
\begin{theorem}
\label{prop:coveragePositiveInverse}
Given clauses $C$ and $D$ such that every repair literal in $C$ and $D$ corresponds to an MD, if 
$C \models D$, $C$ $\theta$-subsumes $D$.
\end{theorem}
\vspace{-5pt}
\noindent
We leverage Theorem~\ref{prop:coveragePositiveInverse} to check whether $C$ covers $e$ efficiently as follows.
Let $C^{md}$ and $G^{md}_{e}$ be the clauses that have the same head literal as $C$ and $G_{e}$ and 
contain all body literals in $C$ and $G_e$ without any connected repair literal and the ones where
all their connected repair literals correspond to some MDs, respectively.
Thus, if there is no subsumption between $C$ and $G_e$, our algorithm tries to find a subsumption between
$C^m$ and $G^{m}_{e}$. If there is no subsumption mapping between $C^m$ and $G^{m}_{e}$,   
$C$ does not cover $e$. Otherwise, let $C^{cfd}$ and $G^{cfd}_{e}$ be the set of body literals of $C$ and $G_e$
that do not appear in the body of $C^{md}$ and $G^{md}_{e}$, respectively.
We apply the repair literals in $C^{cfd}$ and $G^{cfd}_{e}$ in $C$ and $D$ and perform 
subsumption checking for pairs of resulting clauses. If every obtained clause of $C$ $\theta$-subsumes
at least one resulting clause of $G_e$, $C$ covers $e$. Otherwise, $C$ does not cover $e$.
We note than the resulting clauses are not repairs of $C$ and $G_e$ as they sill have the 
repair literals that correspond to some MD.

We follow a similar method to the one explained in the preceding paragraph to check whether clause $C$ covers a negative example with the difference that we use the semantic introduced in Definition~\ref{def:coverExampleNegative} to determine the coverage of negative examples.
Let $G_{e^{-}}$ be the ground bottom-clause for the negative example $e^{-}$.
We generate all repaired clauses of the clause $C$ as described in Section~\ref{section:approach-explanation}.
Then, we check whether each repaired clause of $C$ $\theta$-subsumes $G_{e^{-}}$
the same way as checking $\theta$-subsumption for $C$ and a ground bottom-clause for a positive example.
$C$ $\theta$-subsumes $G_{e^{-}}$ as soon as one repaired clause of $C$ $\theta$-subsumes 
$G_{e^{-}}$. 
\begin{proposition}
\label{prop:coverageNegative}
Given the clause $C$ and ground bottom-clause $G_{e^{-}}$ for negative example $e^{-}$ relative to database $I$, 
clause $C$ covers $e^{-}$ iff a repair of $C$ $\theta$-subsumes $G_{e}^{-}$.
\end{proposition}
\noindent

\label{section:heterolearn-exchange-cleaning-sampling}
{\bf Commutativity of Cleaning \& Learning:}
An interesting question is whether our algorithm produces essentially the same answer as the one that 
learns a repaired definition over each repair of $I$ separately.
We show that, roughly speaking, our algorithm delivers the same information as the one that separately learns over each repaired instance. Thus, our algorithm learns using 
the compact representation without any loss of information.
Let $\mathit{RepairedCls}(C)$ denote the set of all repaired clauses of clause $C$.
Let {\it BC($e$, $I, \Sigma$, $\Phi$)} denote the bottom-clause generated by 
applying the bottom-clause construction algorithm in Section~\ref{section:bottom-clause-construction} 
using example $e$ over database $I$ with the set of MDs $\Sigma$ and CFDs $\Phi$. 
Also, let {\it BC$_{r}$($e$, $\mathit{RepairedInst}(I, \Sigma, \Phi)$)} be the set of 
repaired clauses generated by applying the bottom-clause construction to each repair of $I$ for $e$. 
\begin{theorem}
\label{lemma:BC}
	Given database $I$ with MDs $\Sigma$, CFDs $\Phi$ and set of positive examples 
	$E^+$, for every positive example $e \in E^+$
	$\mathit{BC_{r}}(e, \mathit{RepairedInst}(I, \Sigma, \Phi)) = 
	\mathit{RepairedCls}(\mathit{BC}(e, I, \Sigma))$.
\end{theorem}
\noindent
Now, assume that {\it Generalize($C, e',$ $I, \Sigma,$ $\Phi$)} denotes the clause produced by generalizing $C$ 
to cover example $e'$ over database $I$ with the set of MDs $\Sigma$ and CFDs $\Phi$ 
in a single step of applying the 
algorithm in Section~\ref{sec:generalization}. 
Give a set of repaired clauses ${\bf C}$, let $\mathit{Generalize_{r}}$ $({\bf C}, e', \mathit{RepairedInst}$ 
$(I, \Sigma, \Phi))$ be the set of repaired clauses produced by generalizing 
every repaired clause in ${\bf C}$ to cover example $e'$ in some repair of $I$ using
the algorithm in Section~\ref{sec:generalization}.
\begin{theorem}
	\label{lemma:GN}
	Given database $I$ with MDs $ \Sigma$ and set of positive examples $E^+$
	$\mathit{Generalize_{r}}(StableCls(C), e', \mathit{RepairedInst}(I,$ $\Sigma, \Phi))$ $= 
	 \mathit{RepairedCls}(\mathit{Generalize}(C, I, e', \Sigma, \Phi)).$
\end{theorem}
\vspace{-5pt}
\noindent

		\section{Implementation}
\label{section:heterolearn-implementation-details}
DLearn is implemented on top of VoltDB, {\it voltdb.com}, a main-memory RDBMS. We use the indexing and query processing mechanisms of the database system to create the (ground) bottom-clauses efficiently. 
The set of tuples $I_{e}$ that DLearn gathers to build a bottom-clause may be large if many tuples in $I$ are relevant to $e$, particularly when learning over a large database.
To overcome this problem, DLearn randomly samples from the tuples in $I_{e}$ to obtain a smaller tuple set $I_{e}^s \subseteq I_{e}$
and crates the bottom-clause based on the sampled data ~\cite{progolem,castor:SIGMOD17}..
To do so, DLearn restricts the number of literals added to the bottom-clause per relation through a parameter called {\it sample size}.
To implement similarity over strings, DLearn uses the operator defined as the average of the {\it Smith-Waterman-Gotoh} and the {\it Length} similarity functions.
The {\it Smith-Waterman-Gotoh} function~\cite{Gotoh1982AnIA} measures the similarity of two strings based on their local sequence alignments.
The {\it Length} function computes the similarity of the length of two strings by dividing the length of the smaller string by the length of the larger string.
To improve efficiency, we precompute the pairs of similar values.

		\section{Experiments}
\label{section:heterolearn-experiments}
\subsection{Experimental Settings}
\subsubsection{Datasets}

\begin{table}
	\centering
	\caption{{\small  Numbers of relations (\#R), tuples (\#T), positive examples (\#P), and negative examples (\#N) for each dataset.}}
	\vspace{-10pt}
	\begin{tabular} { c | c | c | c | c }
		\hline
		Name & \#R & \#T & \#P & \#N \\
		\hline
		IMDB & 9 & 3.3M & \multirow{2}{*}{100} & \multirow{2}{*}{200} \\
		OMDB & 15 & 4.8M & &  \\
		\hline
		Walmart & 8 & 19K & \multirow{2}{*}{77} & \multirow{2}{*}{154} \\
		Amazon & 13 & 216K & & \\
		\hline
		DBLP & 4 & 15K & \multirow{2}{*}{500} & \multirow{2}{*}{1000} \\
		Google Scholar & 4 & 328K & & \\
	\end{tabular}
	\label{table:datasets-statistics}
\end{table}

We use databases shown in Table~\ref{table:datasets-statistics}.

\noindent
{\bf IMDB + OMDB}:
The Internet Movie Database (IMDB) and Open Movie Database (OMDB) contain information about movies, such as their titles, year and country of production, genre, directors, and actors~\cite{magellandata}. 
We learn the target relation {\it dramaRestrictedMovies(imdbId)}, which contains the {\it imdbId} of movies that are of the {\it drama} genre and are rated R. 
The {\it imdbId} is only contained in the IMDB database, the genre information is contained in both databases, and the rating information is only contained in the OMDB database. 
We specify an MD that matches movie titles in IMDB with movie titles in OMDB.
We refer to this dataset with one MD as {\bf IMDB + OMDB (one MD)}.
We also create MDs that match cast members and writer names between the two databases. We refer to the dataset that contains the three MDs as {\bf IMDB + OMDB (three MDs)}.

\noindent
{\bf Walmart + Amazon}:
The Walmart and Amazon databases contain information about products, such as their brand, price, categories, dimensions, and weight~\cite{magellandata}.
We learn the target relation {\it upcOfComputersAccessories(upc)}, which contains the {\it upc} of products that are of category {\it Computers Accessories}. 
The {\it upc} is contained in the Walmart database and the information about categories of products is contained in the Amazon database. 
We use an MD that connects the product names across the datasets. 

\noindent
{\bf DBLP + Google Scholar}:
The DBLP and Google Scholar databases contain information about academic papers, such as their titles, authors, and venue and year of publication~\cite{magellandata}. 
The information in the Google Scholar database is not clean, complete, or consistent, e.g., many tuples are missing the year of publication.
Therefore, we aim to augment the information in the Google Scholar database with information from the DBLP database.
We learn the target relation {\it gsPaperYear(gsId, year)}, which contains the Google Scholar id {\it gsId} and the {\it year} of publication of the paper as indicated in the DBLP database.
We use two MDs that match titles and venues in datasets.
\subsubsection{CFDs}
We find 4, 6, and 2 CFDs for IMDB+OMDB, Amazon+Walmart, and
DBLP+Google Scholar, respectively, e.g., 
{\it id} determines {\it title} in Google Scholar.
To test the performance of DLearn on data that contains CFD violations, we inject each aforementioned dataset with varying proportions of CFD violations, $p$, randomly. For example,
$p$ of 5\% means that 5\% of tuples in each relation  violate at least one CFD.

\subsubsection{Systems, Metrics, and Environment}
We compare DLearn against three baseline methods to evaluate the 
handling of MDs over datasets with only MDs. 
These methods use {\it Castor}, a state-of-the-art relational learning system \cite{castor:SIGMOD17}.

\noindent
{\bf Castor-NoMD}: We use Castor to learn over the original databases. It does not use any information from MDs.

\noindent
{\bf Castor-Exact}: We use Castor, but allow the attributes that appear in an MD to be joined through exact joins. Therefore, this system uses information from MDs but only considers exact matches between values.
	
\noindent
{\bf Castor-Clean}: We resolve the heterogeneities between entity names in attributes that appear in an MD by matching each entity in one database with the most similar entity in the other database. 
We use the same similarity function used by DLearn. 
Once the entities are resolved, we use Castor to learn over the unified and clean database.

To evaluate the effectiveness and efficiency of the version of 
DLearn that supports both MDs and CFDs, {\bf DLearn-CFD} 
we compare it with a version of DLearn that supports only MDs and is run over a version of the database whose CFD violations are repaired,
{\bf DLearn-Repaired}. 
We obtain this repair using the minimal repair method, which is popular in repairing CFDs \cite{DBLP:series/synthesis/2012Fan}.
This enables us to evaluate our method for each 
type of inconsistencies separately.	
We perform 5-fold cross validation over all datasets and report the average F1-score and time over the cross validation.
DLearn uses the parameter {\it sample size} to restrict the size of (ground) bottom-clauses. We fix {\it sample size} to 10. 
All systems use 16 threads to parallelize coverage testing.
We use a server with 30 2.3GHz Intel Xeon E5-2670 processors, running CentOS Linux with 500GB of main memory.

\subsection{Empirical Results}
\label{section-experiments-learning}
\begin{table*}
		\centering
		\caption{Results of learning over all datasets with MDs. Number of top similar matches denoted by $k_m$.}
		     \vspace{-10pt}
			\begin{tabular} { c|c|c|c|c|c|c|c }
				\hline
				\multirow{2}{*}{Dataset} & \multirow{2}{*}{Metric} & Castor- & Castor- & Castor- & \multicolumn{3}{c}{DLearn}  \\
				\cline{6-8}
				& & NoMD & Exact & Clean & $k_m=2$ & $k_m=5$ & $k_m=10$ \\
				\hline
				IMDB + OMDB & F1-score & 0.47 & 0.59 & 0.86 & 0.90 & {\bf 0.92} & {\bf 0.92} \\
				(one MD) & Time (m) & 0.12 & 0.13 & 0.18 & 0.26 & 0.42 & 0.87 \\
				\hline
				\hline
				IMDB + OMDB & F1-score & 0.47 & 0.82 & 0.86 & 0.90 & {\bf 0.93} & 0.89 \\
				(three MDs) & Time (m) & 0.12 & 0.48 & 0.21 & 0.30 & 25.87 & 285.39 \\
				\hline
				\hline
				Walmart + & F1-score & 0.39 & 0.39 & 0.61 & 0.61 & 0.63 & {\bf 0.71} \\
				Amazon & Time (m) & 0.09 & 0.13 & 0.13 & 0.13 & 0.13 & 0.17 \\
				\hline
				\hline
				DBLP + & F1-score & 0 & 0.54 & 0.61 & 0.67 & 0.71 & {\bf 0.82} \\
				Google Scholar & Time (m) & 2.5 & 2.5 & 3.1 & 2.7 & 2.7 & 2.7 \\
		\end{tabular}
		\label{table:results-all}
\end{table*}
\begin{table*}
    \centering
    \caption{Results of learning over all datasets with MDs and CFD violations. $p$ is the percentage of CFD violation.}
         \vspace{-10pt}
        \begin{tabular}{c|c|c|c|c|c|c|cc}
            \hline
    				\multirow{2}{*}{Dataset} & \multirow{2}{*}{Metric} &
    				\multicolumn{3}{c|}{DLearn-CFD} &
    				\multicolumn{3}{c}{DLearn-Repaired}  \\
    				\cline{3-8}
    				& & $p=0.05$ & $p=0.10$ & $p=0.20$ & $p=0.05$ & $p=0.10$ & $p=0.20$ \\
    				\hline
    		\multirow{2}{*}{\begin{tabular}[c]{@{}c@{}}IMDB + \\ OMDB (three MDs)\end{tabular}} & F-1 Score & {\bf0.79} & {\bf0.78} & {\bf0.73} & 0.76 & 0.73 & 0.50 \\
             & Time (m) & 11.15 & 16.26 & 26.95 & 5.70 & 12.54 & 22.28 \\
             \hline
             \hline
            \multirow{2}{*}{\begin{tabular}[c]{@{}c@{}}Walmart + \\ Amazon\end{tabular}} & F-1 Score & {\bf0.64} & {\bf0.61} & 0.54 & 0.49 & 0.52 & {\bf0.56} \\
             & Time (m) & 0.17 & 0.2 & 0.23 & 0.18 & 0.18 & 0.19 \\
             \hline
             \hline
            \multirow{2}{*}{\begin{tabular}[c]{@{}c@{}}DBLP + \\ Google Scholar\end{tabular}} & F-1 Score & {\bf0.79} & {\bf0.68} & {\bf0.47} & 0.73 & 0.55 & 0.23 \\
             & Time (m) & 5.92 & 7.04 & 8.57 & 2.51 & 2.6 & 6.51
        \end{tabular}
    \label{table:results-fds}
\end{table*}
\begin{table*}
    \centering
     \caption{\small Learning  over the IMDB+OMDB (3 MDs) with CFD violations by increasing positive (\#P) and negative (\#N) examples.}
     \vspace{-10pt}
    \begin{tabular}{c|c|c|c|c|c|c|c|cc}
        \hline
        \multirow{2}{*}{\#P/\#N} & \multicolumn{4}{c|}{$k_m=5$} & \multicolumn{4}{c}{$k_m=2$} \\
        \cline{2-7}
         & 100/200 & 500/1k & 1k/2k& 2k/4k & 100/200 & 500/1k & 1k/2k & 2k/4k \\
         \hline 
        F-1 Score &0.78 & 0.82 & 0.81 & 0.82 & 0.78 & 0.79 & 0.81 & 0.81 \\
        Time (m) &16.26 & 72.16 & 121.04 & 317.5 & 0.34 & 2.01 & 2.76 & 5.19
    \end{tabular}
\label{table:results-fds-scalability-2}
\end{table*}
\begin{figure*} 
	\centering
	\includegraphics[width=0.3\linewidth]{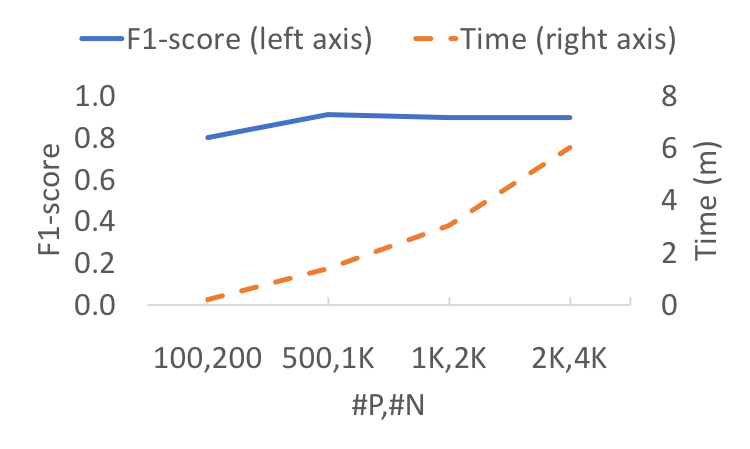}
	\includegraphics[width=0.3\linewidth]{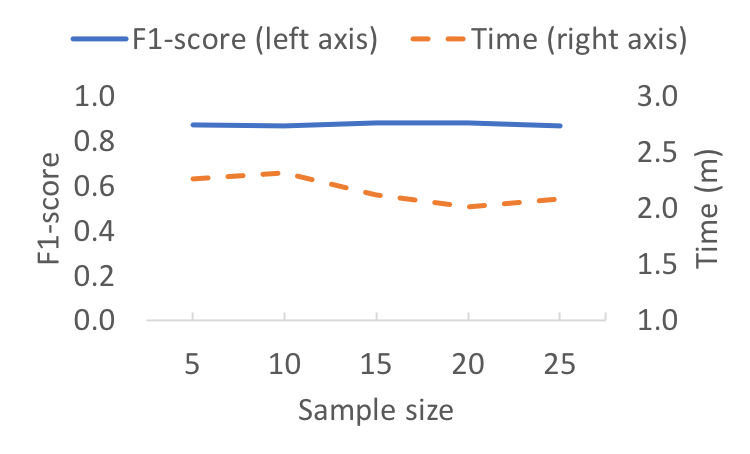}
	\includegraphics[width=0.3\linewidth]{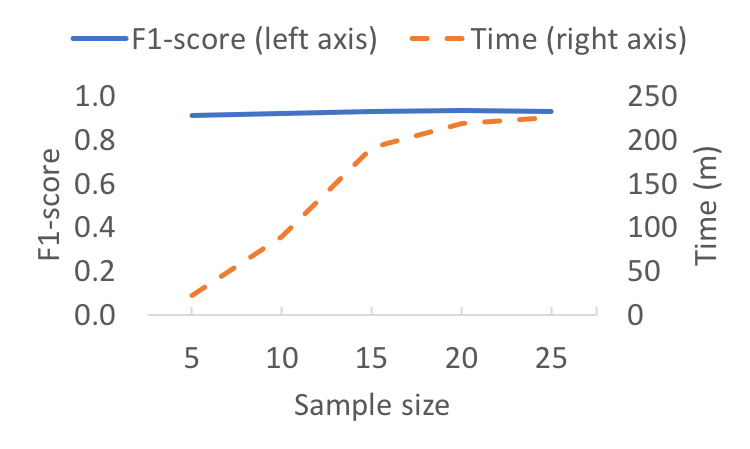}
	\vspace{-10pt}
	\caption{{\small Learning over the IMDB+OMDB (3 MDs) dataset while increasing the number of positive and negative (\#P, \#N) examples (left) and while increasing sample size for $k_m=2$ (middle) and $k_m=5$ (right).}}
	\label{figure:results-scalability}
\end{figure*}
\subsubsection{Handling MDs}
Table~\ref{table:results-all} shows the results over all datasets using DLearn and the baseline systems.
DLearn obtains a better F1-score than the baselines for all datasets.
Castor-Exact obtains a competitive F1-score in the IMDB + BOM dataset with three MDs. 
The MDs that match cast members and writer names between the two databases contain many exact matches. 
DLearn also learns effective definitions over heterogeneous databases efficiently.
Using MDs enables DLearn to consider more patterns, thus, learn a more effective definition.
For example, Castor-Clean learns the following definition over Walmart + Amazon:
{\footnotesize
\begin{align*}
\mathit{u}&\mathit{pcComputersAccessories}(v0) \leftarrow \mathit{walmart\_ids}(v1,v2,v0),\\ 
&\mathit{walmart\_title}(v1,v9), v9 = v10,\\
&\mathit{walmart\_groupname}(v1,\mathit{``Electronics - General"}), \\ 
& \mathit{amazon\_title}(v11,v10),  \mathit{amazon\_listprice}(v11,v16).\\ 
&\text{(positove covered=29, negative covered=11)}\\
\mathit{u}&\mathit{pcComputersAccessories}(v0) \leftarrow \mathit{walmart\_ids}(v1,v2,v0),\\ 
&\mathit{walmart\_title}(v1,v6), v6 = v7, \mathit{amazon\_title}(v8,v7),\\\
& \mathit{amazon\_category}(v8,\mathit{``Computers Accessories"}).\\
& \text{(positove covered=38, negative covered=4)}
\end{align*}
}
The definitions learned by DLearn over the same data is:
{\footnotesize
\begin{align*}
\mathit{u}&\mathit{pcComputersAccessories}(v0) \leftarrow \mathit{walmart\_ids}(v1,v2,v0),\\ 
& \mathit{walmart\_title}(v1,v9), v9 \approx v10,\\
& \mathit{amazon\_title}(v11,v10), \mathit{amazon\_itemweight}(v11,v16),\\
& \mathit{amazon\_category}(v11,\mathit{``Computers Accessories"}). \\
& \text{(positove covered=35, negative covered=5)}\\
\mathit{u}&\mathit{pcComputersAccessories}(v0) \leftarrow
\mathit{walmart\_ids}(v1,v2,v0), \\
& \mathit{walmart\_brand}(v1,\mathit{``Tribeca"}).\\
& \text{(positove covered=8, negative covered=0)}
\end{align*}
}
The definition learned by DLearn has higher precision; they have a similar recall.  
Castor-Clean first learns a clause that covers many positive examples but is not the desired clause. 
This affects its precision. DLearn first learns the desired clause and then learns a clause that has 
high precision.

The effectiveness of the definitions learned by DLearn depends on the number of matches considered in MDs, denoted by $k_m$.
In the Walmart + Amazon, IMDB + BOM (one MD), and DBLP + Google Scholar datasets, using a higher $k_m$ value results in learning a definition with higher F1-score.
When using multiple MDs or when learning a difficult concept, a high $k_m$ value affects DLearn's effectiveness.
In these cases, incorrect matches represent noise that affects DLearn's ability to learn an effective definition.
Nevertheless, it still delivers a more effective definition that other methods.
As the value of $k_m$ increases so does the learning time. 
This is because DLearn has to process more information.

Next, we evaluate the effect of sampling on DLearn's effectiveness and efficiency.
We use the IMDB + OMDB (three MDs) dataset and fix $k_m=2$ and $k_m=5$. We use 800 positive and 1600 negative examples for training, and 200 positive and 400 negative examples for testing.
Figure~\ref{figure:results-scalability} (middle and right) shows the F1-score and learning time of DLearn with $k_m=2$ and $k_m=5$, respectively, when varying the sample size.
For both values of $k_m$, the F1-score does {\it not} change significantly with different sampling sizes.
With $k_m=2$, the learning time remains almost the same with different sampling sizes.
However, with $k_m=5$, the learning time increases significantly.
Therefore, using a small sample size is enough for learning an effective definition efficiently.

\subsubsection{Handling MDs and CFDs}
Table~\ref{table:results-fds} compares DLearn-Repaired and DLearn-CFD. 
Over all three datasets DLearn-CFD performs (almost) equal to or substantially better than the baseline at all levels of violation injection. 
Since DLearn-CFD learns over all possible repairs of violating tuples, it has more available information and consequently 
its hypothesis space is a super-set of the one used by DLearn-Repaired.
In most datasets, the difference is more significant as the proportion
of violations increase. Both methods deliver less effective results when there are more CFD violations in the data. However, 
DLearn-CFD is still able to deliver reasonably effective definitions.
We use $k_m=10$ for DBLP+Google Scholar and Amazon+Walmart and $k_m=5$ for IMDB+OMDB as it takes a long time to use $k_m=5$ for the latter.

\subsubsection{Impact of Number of Iterations}
\label{section-experiments-iterations}
\begin{table}
    \centering
    \caption{Results of changing the number of iterations.}
    \vspace{-10pt}
    \begin{tabular}{c|c|c|c|cc}
        \hline
        \multirow{2}{*}{Metric} & \multicolumn{4}{c}{$k_m=5$} \\
        \cline{2-5}
         & d=2 & d=3 & d=4 & d=5 \\
         \hline 
        F-1 Score & 0.52 & 0.52 & 0.78 & 0.80 \\
        Time (m) & 1.35 & 4.35 & 16.26 & 37.56
    \end{tabular}
    \label{table:results-fds-scalability}
\vspace{-15pt}
\end{table}

We have used values 3, 4, and 5 for the number of iterations, $d$, 
for DBLP+Google Scholar, IMDb+OMDB, and Walmart+Amazon datasets, respectively.
Table~\ref{table:results-fds-scalability} shows data regarding the scalability of DLearn-CFD over IMDb+OMDB (3 MD + 4 CFD). 
A higher $d$-value increases both the effectiveness as well as the runtime. 
We fix the value $k_m$ at 5. 
A $d$-value higher than 4 generates a very modest increase in effectiveness with a substantial increase in runtime. 
This result indicates that for a given dataset, the learning algorithm can access most relevant tuples for a reasonable value of $d$. 

\subsubsection{Scalability of DLearn}
\label{section-experiments-scalability-examples}
We evaluate the effect of the number of training examples in both DLearn's effectiveness and efficiency.
We use the IMDB + OMDB (three MDs) dataset and fix $k_m=2$.
We generate 2100 positive and 4200 negative examples.
From these sets, we use 100 positive and 200 negative examples for testing.
From the remaining examples, we generate training sets containing 100, 500, 1000, and 2000 positive examples, and double the number of negative examples.
For each training set, we use DLearn with MD support 
to learn a definition.
Figure~\ref{figure:results-scalability} (left) shows the F1-scores and learning times for each training set.
With 100 positive and 200 negative examples, DLearn obtains an F1-score of 0.80.
With 500 positive and 1000 negative examples, the F1-score increases to 0.91.
DLearn is able to learn efficiently even with the largest training set.
We also evaluate DLearn with support for both MDs and CFDs' violations
and report the results in Table~\ref{table:results-fds-scalability-2}.
It indicate that DLearn with CFD and MD support can 
deliver effective results efficiently over large number of examples with $k_m=2$.

		\section{Related Work}
\label{section:heterolearn-related}

Data cleaning is an important and flourishing area in database research \cite{Fan2009ReasoningAR,DBLP:conf/sigmod/ChuIKW16,DBLP:conf/kr/BahmaniBKL12,DBLP:series/synthesis/2012Fan,Burdick:2016:DFL:2966276.2894748}.
Most data cleaning systems leverage declarative constraints to produce clean instances.

ActiveClean gradually cleans a dirty dataset to learn a convex-loss model, such as Logistic Regression \cite{Krishnan:2016:AID:2882903.2899409}. Its goal is to clean the underlying dataset such that the learned model becomes more effective as it receives more cleaned records potentially from the user.
Our objective, however, is to learn a model over the original data without cleaning it. 
Furthermore, ActiveClean does {\it not} address the problem of having multiple cleaned instances.

		\section{Conclusion \& Future Work}
\label{section:heterolearn-conclusion}
We investigated the problem of learning directly over heterogeneous data and proposed a new method that leverages constraints in learning to represent inconsistencies.
Since most of these quality problems have been modeled using declarative constraints, we 
plan to extend our framework to address more quality issues.

		\bibliographystyle{ACM-Reference-Format}
		\bibliography{./ref}


\begin{thebibliography}{54}


\ifx \showCODEN    \undefined \def \showCODEN     #1{\unskip}     \fi
\ifx \showDOI      \undefined \def \showDOI       #1{#1}\fi
\ifx \showISBNx    \undefined \def \showISBNx     #1{\unskip}     \fi
\ifx \showISBNxiii \undefined \def \showISBNxiii  #1{\unskip}     \fi
\ifx \showISSN     \undefined \def \showISSN      #1{\unskip}     \fi
\ifx \showLCCN     \undefined \def \showLCCN      #1{\unskip}     \fi
\ifx \shownote     \undefined \def \shownote      #1{#1}          \fi
\ifx \showarticletitle \undefined \def \showarticletitle #1{#1}   \fi
\ifx \showURL      \undefined \def \showURL       {\relax}        \fi
\providecommand\bibfield[2]{#2}
\providecommand\bibinfo[2]{#2}
\providecommand\natexlab[1]{#1}
\providecommand\showeprint[2][]{arXiv:#2}

\bibitem[\protect\citeauthoryear{Abedjan, Golab, and Naumann}{Abedjan
  et~al\mbox{.}}{2015}]%
        {Abedjan:2015:PRD:2811716.2811766}
\bibfield{author}{\bibinfo{person}{Ziawasch Abedjan}, \bibinfo{person}{Lukasz
  Golab}, {and} \bibinfo{person}{Felix Naumann}.}
  \bibinfo{year}{2015}\natexlab{}.
\newblock \showarticletitle{Profiling relational data: a survey}.
\newblock \bibinfo{journal}{\emph{The VLDB Journal}}  \bibinfo{volume}{24}
  (\bibinfo{year}{2015}), \bibinfo{pages}{557--581}.
\newblock


\bibitem[\protect\citeauthoryear{Abiteboul, Hull, and Vianu}{Abiteboul
  et~al\mbox{.}}{1994}]%
        {AliceBook}
\bibfield{author}{\bibinfo{person}{Serge Abiteboul}, \bibinfo{person}{Richard
  Hull}, {and} \bibinfo{person}{Victor Vianu}.}
  \bibinfo{year}{1994}\natexlab{}.
\newblock \bibinfo{booktitle}{\emph{{Foundations of Databases: The Logical
  Level}}}.
\newblock \bibinfo{publisher}{Addison-Wesley}.
\newblock


\bibitem[\protect\citeauthoryear{Abouzeid, Angluin, Papadimitriou, Hellerstein,
  and Silberschatz}{Abouzeid et~al\mbox{.}}{2013}]%
        {Abouzied:PODS:13}
\bibfield{author}{\bibinfo{person}{Azza Abouzeid}, \bibinfo{person}{Dana
  Angluin}, \bibinfo{person}{Christos~H. Papadimitriou},
  \bibinfo{person}{Joseph~M. Hellerstein}, {and} \bibinfo{person}{Abraham
  Silberschatz}.} \bibinfo{year}{2013}\natexlab{}.
\newblock \showarticletitle{Learning and verifying quantified boolean queries
  by example}. In \bibinfo{booktitle}{\emph{PODS}}.
\newblock


\bibitem[\protect\citeauthoryear{Afrati, Li, and Mitra}{Afrati
  et~al\mbox{.}}{2004}]%
        {10.1007/978-3-540-24741-8_27}
\bibfield{author}{\bibinfo{person}{Foto Afrati}, \bibinfo{person}{Chen Li},
  {and} \bibinfo{person}{Prasenjit Mitra}.} \bibinfo{year}{2004}\natexlab{}.
\newblock \showarticletitle{On Containment of Conjunctive Queries with
  Arithmetic Comparisons}. In \bibinfo{booktitle}{\emph{Advances in Database
  Technology - EDBT}}. \bibinfo{pages}{459--476}.
\newblock


\bibitem[\protect\citeauthoryear{Arasu, Chaudhuri, and Kaushik}{Arasu
  et~al\mbox{.}}{2008}]%
        {4497412}
\bibfield{author}{\bibinfo{person}{Arvind Arasu}, \bibinfo{person}{Surajit
  Chaudhuri}, {and} \bibinfo{person}{Raghav Kaushik}.}
  \bibinfo{year}{2008}\natexlab{}.
\newblock \showarticletitle{Transformation-based Framework for Record
  Matching}.
\newblock \bibinfo{journal}{\emph{ICDE}} (\bibinfo{year}{2008}),
  \bibinfo{pages}{40--49}.
\newblock


\bibitem[\protect\citeauthoryear{Arenas, Bertossi, and Chomicki}{Arenas
  et~al\mbox{.}}{1999}]%
        {Arenas:PODS:99:Consistent}
\bibfield{author}{\bibinfo{person}{Marcelo Arenas},
  \bibinfo{person}{Leopoldo~E. Bertossi}, {and} \bibinfo{person}{Jan
  Chomicki}.} \bibinfo{year}{1999}\natexlab{}.
\newblock \showarticletitle{Consistent Query Answers in Inconsistent
  Databases}. In \bibinfo{booktitle}{\emph{PODS}}.
\newblock


\bibitem[\protect\citeauthoryear{Bahmani, Bertossi, Kolahi, and
  Lakshmanan}{Bahmani et~al\mbox{.}}{2012}]%
        {DBLP:conf/kr/BahmaniBKL12}
\bibfield{author}{\bibinfo{person}{Zeinab Bahmani},
  \bibinfo{person}{Leopoldo~E. Bertossi}, \bibinfo{person}{Solmaz Kolahi},
  {and} \bibinfo{person}{Laks V.~S. Lakshmanan}.}
  \bibinfo{year}{2012}\natexlab{}.
\newblock \showarticletitle{Declarative Entity Resolution via Matching
  Dependencies and Answer Set Programs}. In \bibinfo{booktitle}{\emph{KR}}.
\newblock


\bibitem[\protect\citeauthoryear{Bahmani, Bertossi, and Vasiloglou}{Bahmani
  et~al\mbox{.}}{2015}]%
        {Bahmani2015ERBloxCM}
\bibfield{author}{\bibinfo{person}{Zeinab Bahmani},
  \bibinfo{person}{Leopoldo~E. Bertossi}, {and} \bibinfo{person}{Nikolaos
  Vasiloglou}.} \bibinfo{year}{2015}\natexlab{}.
\newblock \showarticletitle{{ERBlox}: Combining Matching Dependencies with
  Machine Learning for Entity Resolution}. In \bibinfo{booktitle}{\emph{SUM}}.
\newblock


\bibitem[\protect\citeauthoryear{Benjelloun, Garcia-Molina, Menestrina, Su,
  Whang, and Widom}{Benjelloun et~al\mbox{.}}{2008}]%
        {Benjelloun:2009:SGA:1541533.1541538}
\bibfield{author}{\bibinfo{person}{Omar Benjelloun}, \bibinfo{person}{Hector
  Garcia-Molina}, \bibinfo{person}{David Menestrina}, \bibinfo{person}{Qi Su},
  \bibinfo{person}{Steven~Euijong Whang}, {and} \bibinfo{person}{Jennifer
  Widom}.} \bibinfo{year}{2008}\natexlab{}.
\newblock \showarticletitle{Swoosh: a generic approach to entity resolution}.
\newblock \bibinfo{journal}{\emph{The VLDB Journal}}  \bibinfo{volume}{18}
  (\bibinfo{year}{2008}), \bibinfo{pages}{255--276}.
\newblock


\bibitem[\protect\citeauthoryear{Bertossi, Kolahi, and Lakshmanan}{Bertossi
  et~al\mbox{.}}{2011}]%
        {Bertossi2011DataCA}
\bibfield{author}{\bibinfo{person}{Leopoldo~E. Bertossi},
  \bibinfo{person}{Solmaz Kolahi}, {and} \bibinfo{person}{Laks V.~S.
  Lakshmanan}.} \bibinfo{year}{2011}\natexlab{}.
\newblock \showarticletitle{Data Cleaning and Query Answering with Matching
  Dependencies and Matching Functions}.
\newblock \bibinfo{journal}{\emph{Theory of Computing Systems}}
  \bibinfo{volume}{52} (\bibinfo{year}{2011}), \bibinfo{pages}{441--482}.
\newblock


\bibitem[\protect\citeauthoryear{Bohannon, Fan, Flaster, and Rastogi}{Bohannon
  et~al\mbox{.}}{2005}]%
        {10.1145/1066157.1066175}
\bibfield{author}{\bibinfo{person}{Philip Bohannon}, \bibinfo{person}{Wenfei
  Fan}, \bibinfo{person}{Michael Flaster}, {and} \bibinfo{person}{Rajeev
  Rastogi}.} \bibinfo{year}{2005}\natexlab{}.
\newblock \showarticletitle{A Cost-Based Model and Effective Heuristic for
  Repairing Constraints by Value Modification}. In
  \bibinfo{booktitle}{\emph{Proceedings of the 2005 ACM SIGMOD International
  Conference on Management of Data}}. \bibinfo{publisher}{Association for
  Computing Machinery}, \bibinfo{address}{New York, NY, USA},
  \bibinfo{pages}{143?154}.
\newblock
\showISBNx{1595930604}
\urldef\tempurl%
\url{https://doi.org/10.1145/1066157.1066175}
\showDOI{\tempurl}


\bibitem[\protect\citeauthoryear{{Bohannon}, {Fan}, {Geerts}, {Jia}, and
  {Kementsietsidis}}{{Bohannon} et~al\mbox{.}}{2007}]%
        {4221723}
\bibfield{author}{\bibinfo{person}{P. {Bohannon}}, \bibinfo{person}{W. {Fan}},
  \bibinfo{person}{F. {Geerts}}, \bibinfo{person}{X. {Jia}}, {and}
  \bibinfo{person}{A. {Kementsietsidis}}.} \bibinfo{year}{2007}\natexlab{}.
\newblock \showarticletitle{Conditional Functional Dependencies for Data
  Cleaning}. In \bibinfo{booktitle}{\emph{2007 IEEE 23rd International
  Conference on Data Engineering}}. \bibinfo{pages}{746--755}.
\newblock


\bibitem[\protect\citeauthoryear{Burdick, Fagin, Kolaitis, Popa, and
  Tan}{Burdick et~al\mbox{.}}{2016}]%
        {Burdick:2016:DFL:2966276.2894748}
\bibfield{author}{\bibinfo{person}{Douglas Burdick}, \bibinfo{person}{Ronald
  Fagin}, \bibinfo{person}{Phokion~G. Kolaitis}, \bibinfo{person}{Lucian Popa},
  {and} \bibinfo{person}{Wang-Chiew Tan}.} \bibinfo{year}{2016}\natexlab{}.
\newblock \showarticletitle{A Declarative Framework for Linking Entities}.
\newblock \bibinfo{journal}{\emph{ACM Trans. Database Syst.}}
  \bibinfo{volume}{41}, \bibinfo{number}{3}, Article \bibinfo{articleno}{17}
  (\bibinfo{year}{2016}), \bibinfo{numpages}{38}~pages.
\newblock
\showISSN{0362-5915}
\urldef\tempurl%
\url{https://doi.org/10.1145/2894748}
\showDOI{\tempurl}


\bibitem[\protect\citeauthoryear{Chu, Ilyas, Krishnan, and Wang}{Chu
  et~al\mbox{.}}{2016}]%
        {DBLP:conf/sigmod/ChuIKW16}
\bibfield{author}{\bibinfo{person}{Xu Chu}, \bibinfo{person}{Ihab~F. Ilyas},
  \bibinfo{person}{Sanjay Krishnan}, {and} \bibinfo{person}{Jiannan Wang}.}
  \bibinfo{year}{2016}\natexlab{}.
\newblock \showarticletitle{Data Cleaning: Overview and Emerging Challenges}.
  In \bibinfo{booktitle}{\emph{SIGMOD Conference}}.
\newblock


\bibitem[\protect\citeauthoryear{Cong, Fan, Geerts, Jia, and Ma}{Cong
  et~al\mbox{.}}{2007}]%
        {Cong:Gao:2007}
\bibfield{author}{\bibinfo{person}{Gao Cong}, \bibinfo{person}{Wenfei Fan},
  \bibinfo{person}{Floris Geerts}, \bibinfo{person}{Xibei Jia}, {and}
  \bibinfo{person}{Shuai Ma}.} \bibinfo{year}{2007}\natexlab{}.
\newblock \showarticletitle{Improving Data Quality: Consistency and Accuracy.}
\newblock \bibinfo{journal}{\emph{Proc. Int'l Conf. Very Large Data Bases
  (VLDB)}}, \bibinfo{pages}{315--326}.
\newblock


\bibitem[\protect\citeauthoryear{Das, Doan, G.~C., Gokhale, and Konda}{Das
  et~al\mbox{.}}{[n. d.]}]%
        {magellandata}
\bibfield{author}{\bibinfo{person}{Sanjib Das}, \bibinfo{person}{AnHai Doan},
  \bibinfo{person}{Paul~Suganthan G.~C.}, \bibinfo{person}{Chaitanya Gokhale},
  {and} \bibinfo{person}{Pradap Konda}.} \bibinfo{year}{[n. d.]}\natexlab{}.
\newblock \bibinfo{title}{The {M}agellan Data Repository}.
\newblock
  \bibinfo{howpublished}{\url{https://sites.google.com/site/anhaidgroup/projects/data}}.
\newblock


\bibitem[\protect\citeauthoryear{De~Raedt}{De~Raedt}{2010}]%
        {DeRaedt:2010:LRL:1952055}
\bibfield{author}{\bibinfo{person}{Luc De~Raedt}.}
  \bibinfo{year}{2010}\natexlab{}.
\newblock \bibinfo{booktitle}{\emph{Logical and Relational Learning}
  (\bibinfo{edition}{1st} ed.)}.
\newblock \bibinfo{publisher}{Springer Publishing Company, Incorporated}.
\newblock
\showISBNx{3642057489, 9783642057489}


\bibitem[\protect\citeauthoryear{Doan, Halevy, and Ives}{Doan
  et~al\mbox{.}}{2012}]%
        {Doan:2012:PDI:2401764}
\bibfield{author}{\bibinfo{person}{AnHai Doan}, \bibinfo{person}{Alon Halevy},
  {and} \bibinfo{person}{Zachary Ives}.} \bibinfo{year}{2012}\natexlab{}.
\newblock \bibinfo{booktitle}{\emph{Principles of Data Integration}
  (\bibinfo{edition}{1st} ed.)}.
\newblock \bibinfo{publisher}{Morgan Kaufmann Publishers Inc.},
  \bibinfo{address}{San Francisco, CA, USA}.
\newblock
\showISBNx{0124160441, 9780124160446}


\bibitem[\protect\citeauthoryear{Domingos}{Domingos}{2018}]%
        {10.1145/3183713.3199515}
\bibfield{author}{\bibinfo{person}{Pedro Domingos}.}
  \bibinfo{year}{2018}\natexlab{}.
\newblock \showarticletitle{Machine Learning for Data Management: Problems and
  Solutions}. In \bibinfo{booktitle}{\emph{Proceedings of the 2018
  International Conference on Management of Data}}
  \emph{(\bibinfo{series}{SIGMOD ?18})}. \bibinfo{publisher}{Association for
  Computing Machinery}, \bibinfo{address}{New York, NY, USA},
  \bibinfo{pages}{629}.
\newblock
\showISBNx{9781450347037}
\urldef\tempurl%
\url{https://doi.org/10.1145/3183713.3199515}
\showDOI{\tempurl}


\bibitem[\protect\citeauthoryear{Eiter, Gottlob, and Mannila}{Eiter
  et~al\mbox{.}}{1997}]%
        {Eiter:1997:DD:261124.261126}
\bibfield{author}{\bibinfo{person}{Thomas Eiter}, \bibinfo{person}{Georg
  Gottlob}, {and} \bibinfo{person}{Heikki Mannila}.}
  \bibinfo{year}{1997}\natexlab{}.
\newblock \showarticletitle{Disjunctive {Datalog}}.
\newblock \bibinfo{journal}{\emph{ACM Trans. Database Syst.}}
  \bibinfo{volume}{22} (\bibinfo{year}{1997}), \bibinfo{pages}{364--418}.
\newblock


\bibitem[\protect\citeauthoryear{Evans and Grefenstette}{Evans and
  Grefenstette}{2018}]%
        {Evans2018LearningER}
\bibfield{author}{\bibinfo{person}{Richard Evans} {and} \bibinfo{person}{Edward
  Grefenstette}.} \bibinfo{year}{2018}\natexlab{}.
\newblock \showarticletitle{Learning Explanatory Rules from Noisy Data}.
\newblock \bibinfo{journal}{\emph{J. Artif. Intell. Res.}}
  \bibinfo{volume}{61} (\bibinfo{year}{2018}), \bibinfo{pages}{1--64}.
\newblock


\bibitem[\protect\citeauthoryear{Fan}{Fan}{2008}]%
        {DBLP:conf/pods/Fan08}
\bibfield{author}{\bibinfo{person}{Wenfei Fan}.}
  \bibinfo{year}{2008}\natexlab{}.
\newblock \showarticletitle{Dependencies revisited for improving data quality}.
  In \bibinfo{booktitle}{\emph{Proceedings of the Twenty-Seventh {ACM}
  {SIGMOD-SIGACT-SIGART} Symposium on Principles of Database Systems, {PODS}
  2008, June 9-11, 2008, Vancouver, BC, Canada}},
  \bibfield{editor}{\bibinfo{person}{Maurizio Lenzerini} {and}
  \bibinfo{person}{Domenico Lembo}} (Eds.). \bibinfo{publisher}{{ACM}},
  \bibinfo{pages}{159--170}.
\newblock
\urldef\tempurl%
\url{https://doi.org/10.1145/1376916.1376940}
\showDOI{\tempurl}


\bibitem[\protect\citeauthoryear{Fan and Geerts}{Fan and Geerts}{2012}]%
        {DBLP:series/synthesis/2012Fan}
\bibfield{author}{\bibinfo{person}{Wenfei Fan} {and} \bibinfo{person}{Floris
  Geerts}.} \bibinfo{year}{2012}\natexlab{}.
\newblock \bibinfo{booktitle}{\emph{Foundations of Data Quality Management}}.
\newblock \bibinfo{publisher}{Morgan {\&} Claypool Publishers}.
\newblock
\urldef\tempurl%
\url{https://doi.org/10.2200/S00439ED1V01Y201207DTM030}
\showDOI{\tempurl}


\bibitem[\protect\citeauthoryear{Fan, Jia, Li, and Ma}{Fan
  et~al\mbox{.}}{2009}]%
        {Fan2009ReasoningAR}
\bibfield{author}{\bibinfo{person}{Wenfei Fan}, \bibinfo{person}{Xibei Jia},
  \bibinfo{person}{Jianzhong Li}, {and} \bibinfo{person}{Shuai Ma}.}
  \bibinfo{year}{2009}\natexlab{}.
\newblock \showarticletitle{Reasoning about Record Matching Rules}.
\newblock \bibinfo{journal}{\emph{PVLDB}}  \bibinfo{volume}{2}
  (\bibinfo{year}{2009}), \bibinfo{pages}{407--418}.
\newblock


\bibitem[\protect\citeauthoryear{Franconi1, Palma, Leone, Perri, and
  Scarcello}{Franconi1 et~al\mbox{.}}{2001}]%
        {10.1007/3-540-45653-8_39}
\bibfield{author}{\bibinfo{person}{Enrico Franconi1},
  \bibinfo{person}{Antonio~Laureti Palma}, \bibinfo{person}{Nicola Leone},
  \bibinfo{person}{Simona Perri}, {and} \bibinfo{person}{Francesco Scarcello}.}
  \bibinfo{year}{2001}\natexlab{}.
\newblock \showarticletitle{Census Data Repair: A Challenging Application of
  Disjunctive Logic Programming}. In \bibinfo{booktitle}{\emph{Logic for
  Programming, Artificial Intelligence, and Reasoning}},
  \bibfield{editor}{\bibinfo{person}{Robert Nieuwenhuis} {and}
  \bibinfo{person}{Andrei Voronkov}} (Eds.). \bibinfo{publisher}{Springer
  Berlin Heidelberg}, \bibinfo{address}{Berlin, Heidelberg},
  \bibinfo{pages}{561--578}.
\newblock


\bibitem[\protect\citeauthoryear{Galhardas, Florescu, Shasha, Simon, and
  Saita}{Galhardas et~al\mbox{.}}{2001}]%
        {Galhardas:2001:DDC:645927.672042}
\bibfield{author}{\bibinfo{person}{Helena Galhardas}, \bibinfo{person}{Daniela
  Florescu}, \bibinfo{person}{Dennis Shasha}, \bibinfo{person}{Eric Simon},
  {and} \bibinfo{person}{Cristian-Augustin Saita}.}
  \bibinfo{year}{2001}\natexlab{}.
\newblock \showarticletitle{Declarative Data Cleaning: Language, Model, and
  Algorithms}. In \bibinfo{booktitle}{\emph{VLDB}}.
\newblock


\bibitem[\protect\citeauthoryear{GarciaMolina, Ullman, and Widom}{GarciaMolina
  et~al\mbox{.}}{2008}]%
        {DBBook}
\bibfield{author}{\bibinfo{person}{Hector GarciaMolina}, \bibinfo{person}{Jeff
  Ullman}, {and} \bibinfo{person}{Jennifer Widom}.}
  \bibinfo{year}{2008}\natexlab{}.
\newblock \bibinfo{booktitle}{\emph{{Database Systems: The Complete Book}}}.
\newblock \bibinfo{publisher}{{Prentice Hall}}.
\newblock


\bibitem[\protect\citeauthoryear{Getoor and Machanavajjhala}{Getoor and
  Machanavajjhala}{2013}]%
        {Getoor:2013:ERB:2487575.2506179}
\bibfield{author}{\bibinfo{person}{Lise Getoor} {and} \bibinfo{person}{Ashwin
  Machanavajjhala}.} \bibinfo{year}{2013}\natexlab{}.
\newblock \showarticletitle{Entity resolution for big data}. In
  \bibinfo{booktitle}{\emph{KDD}}.
\newblock


\bibitem[\protect\citeauthoryear{Getoor and Taskar}{Getoor and Taskar}{2007}]%
        {Getoor:SRLBook}
\bibfield{author}{\bibinfo{person}{Lise Getoor} {and} \bibinfo{person}{Ben
  Taskar}.} \bibinfo{year}{2007}\natexlab{}.
\newblock \bibinfo{booktitle}{\emph{{Introduction to Statistical Relational
  Learning}}}.
\newblock \bibinfo{publisher}{{MIT Press}}.
\newblock


\bibitem[\protect\citeauthoryear{Golab, Karloff, Korn, Srivastava, and
  Yu}{Golab et~al\mbox{.}}{2008}]%
        {10.14778/1453856.1453900}
\bibfield{author}{\bibinfo{person}{Lukasz Golab}, \bibinfo{person}{Howard
  Karloff}, \bibinfo{person}{Flip Korn}, \bibinfo{person}{Divesh Srivastava},
  {and} \bibinfo{person}{Bei Yu}.} \bibinfo{year}{2008}\natexlab{}.
\newblock \showarticletitle{On Generating Near-Optimal Tableaux for Conditional
  Functional Dependencies}.
\newblock \bibinfo{journal}{\emph{Proc. VLDB Endow.}} \bibinfo{volume}{1},
  \bibinfo{number}{1} (\bibinfo{date}{Aug.} \bibinfo{year}{2008}),
  \bibinfo{pages}{376--390}.
\newblock
\showISSN{2150-8097}
\urldef\tempurl%
\url{https://doi.org/10.14778/1453856.1453900}
\showDOI{\tempurl}


\bibitem[\protect\citeauthoryear{Gotoh}{Gotoh}{1982}]%
        {Gotoh1982AnIA}
\bibfield{author}{\bibinfo{person}{Osamu Gotoh}.}
  \bibinfo{year}{1982}\natexlab{}.
\newblock \showarticletitle{An improved algorithm for matching biological
  sequences}.
\newblock \bibinfo{journal}{\emph{Journal of Molecular Biology}}
  \bibinfo{volume}{162 3} (\bibinfo{year}{1982}), \bibinfo{pages}{705--708}.
\newblock


\bibitem[\protect\citeauthoryear{Hern{\'a}ndez, Koutrika, Krishnamurthy, Popa,
  and Wisnesky}{Hern{\'a}ndez et~al\mbox{.}}{2013}]%
        {Hernandez:2013:HHS:2452376.2452440}
\bibfield{author}{\bibinfo{person}{Miguel~{\'A}ngel Hern{\'a}ndez},
  \bibinfo{person}{Georgia Koutrika}, \bibinfo{person}{Rajasekar
  Krishnamurthy}, \bibinfo{person}{Lucian Popa}, {and} \bibinfo{person}{Ryan
  Wisnesky}.} \bibinfo{year}{2013}\natexlab{}.
\newblock \showarticletitle{{HIL}: a high-level scripting language for entity
  integration}. In \bibinfo{booktitle}{\emph{EDBT}}.
\newblock


\bibitem[\protect\citeauthoryear{Hern{\'a}ndez and Stolfo}{Hern{\'a}ndez and
  Stolfo}{1995}]%
        {Hernandez:1995:MPL:223784.223807}
\bibfield{author}{\bibinfo{person}{Mauricio~A. Hern{\'a}ndez} {and}
  \bibinfo{person}{Salvatore~J. Stolfo}.} \bibinfo{year}{1995}\natexlab{}.
\newblock \showarticletitle{The Merge/Purge Problem for Large Databases}. In
  \bibinfo{booktitle}{\emph{SIGMOD Conference}}.
\newblock


\bibitem[\protect\citeauthoryear{Ilyas}{Ilyas}{2016}]%
        {DBLP:journals/debu/Ilyas16}
\bibfield{author}{\bibinfo{person}{Ihab~F. Ilyas}.}
  \bibinfo{year}{2016}\natexlab{}.
\newblock \showarticletitle{Effective Data Cleaning with Continuous
  Evaluation}.
\newblock \bibinfo{journal}{\emph{{IEEE} Data Eng. Bull.}}
  \bibinfo{volume}{39}, \bibinfo{number}{2} (\bibinfo{year}{2016}),
  \bibinfo{pages}{38--46}.
\newblock
\urldef\tempurl%
\url{http://sites.computer.org/debull/A16june/p38.pdf}
\showURL{%
\tempurl}


\bibitem[\protect\citeauthoryear{Kalashnikov, Lakshmanan, and
  Srivastava}{Kalashnikov et~al\mbox{.}}{2018}]%
        {Kalashnikov:2018:FFQ:3183713.3183727}
\bibfield{author}{\bibinfo{person}{Dmitri~V. Kalashnikov},
  \bibinfo{person}{Laks~V.S. Lakshmanan}, {and} \bibinfo{person}{Divesh
  Srivastava}.} \bibinfo{year}{2018}\natexlab{}.
\newblock \showarticletitle{FastQRE: Fast Query Reverse Engineering}. In
  \bibinfo{booktitle}{\emph{Proceedings of the 2018 International Conference on
  Management of Data}} \emph{(\bibinfo{series}{SIGMOD '18})}.
  \bibinfo{publisher}{ACM}, \bibinfo{address}{New York, NY, USA},
  \bibinfo{pages}{337--350}.
\newblock
\showISBNx{978-1-4503-4703-7}
\urldef\tempurl%
\url{https://doi.org/10.1145/3183713.3183727}
\showDOI{\tempurl}


\bibitem[\protect\citeauthoryear{Kimmig, Poole, and Pujara}{Kimmig
  et~al\mbox{.}}{2020}]%
        {StarAIW}
\bibfield{author}{\bibinfo{person}{Angelika Kimmig}, \bibinfo{person}{David
  Poole}, {and} \bibinfo{person}{Jay Pujara}.} \bibinfo{year}{2020}\natexlab{}.
\newblock \showarticletitle{Statistical Relational AI (StarAI) WorkShop}. In
  \bibinfo{booktitle}{\emph{AAAI}}.
\newblock


\bibitem[\protect\citeauthoryear{Kolahi and Lakshmanan}{Kolahi and
  Lakshmanan}{2009}]%
        {10.1145/1514894.1514901}
\bibfield{author}{\bibinfo{person}{Solmaz Kolahi} {and} \bibinfo{person}{Laks
  V.~S. Lakshmanan}.} \bibinfo{year}{2009}\natexlab{}.
\newblock \showarticletitle{On Approximating Optimum Repairs for Functional
  Dependency Violations}. In \bibinfo{booktitle}{\emph{Proceedings of the 12th
  International Conference on Database Theory}}.
  \bibinfo{publisher}{Association for Computing Machinery},
  \bibinfo{address}{New York, NY, USA}, \bibinfo{pages}{53?62}.
\newblock
\showISBNx{9781605584232}
\urldef\tempurl%
\url{https://doi.org/10.1145/1514894.1514901}
\showDOI{\tempurl}


\bibitem[\protect\citeauthoryear{Koumarelas, Papenbrock, and
  Naumann}{Koumarelas et~al\mbox{.}}{2020}]%
        {MDedup:Koumarelas}
\bibfield{author}{\bibinfo{person}{Ioannis Koumarelas},
  \bibinfo{person}{Thorsten Papenbrock}, {and} \bibinfo{person}{Felix
  Naumann}.} \bibinfo{year}{2020}\natexlab{}.
\newblock \showarticletitle{MDedup: Duplicate Detection with Matching
  Dependencies}.
\newblock \bibinfo{journal}{\emph{Proceedings of the VLDB Endowment}}
  \bibinfo{volume}{13}, \bibinfo{number}{5} (\bibinfo{year}{2020}),
  \bibinfo{pages}{712--725}.
\newblock


\bibitem[\protect\citeauthoryear{Krishnan, Wang, Wu, Franklin, and
  Goldberg}{Krishnan et~al\mbox{.}}{2016}]%
        {Krishnan:2016:AID:2882903.2899409}
\bibfield{author}{\bibinfo{person}{Sanjay Krishnan}, \bibinfo{person}{Jiannan
  Wang}, \bibinfo{person}{Eugene Wu}, \bibinfo{person}{Michael~J. Franklin},
  {and} \bibinfo{person}{Kenneth~Y. Goldberg}.}
  \bibinfo{year}{2016}\natexlab{}.
\newblock \showarticletitle{{ActiveClean}: Interactive Data Cleaning For
  Statistical Modeling}.
\newblock \bibinfo{journal}{\emph{PVLDB}}  \bibinfo{volume}{9}
  (\bibinfo{year}{2016}), \bibinfo{pages}{948--959}.
\newblock


\bibitem[\protect\citeauthoryear{Lao, Minkov, and Cohen}{Lao
  et~al\mbox{.}}{2015}]%
        {lao-etal-2015-learning}
\bibfield{author}{\bibinfo{person}{Ni Lao}, \bibinfo{person}{Einat Minkov},
  {and} \bibinfo{person}{William Cohen}.} \bibinfo{year}{2015}\natexlab{}.
\newblock \showarticletitle{Learning Relational Features with Backward Random
  Walks}. In \bibinfo{booktitle}{\emph{Proceedings of the 53rd Annual Meeting
  of the Association for Computational Linguistics and the 7th International
  Joint Conference on Natural Language Processing (Volume 1: Long Papers)}}.
  \bibinfo{publisher}{Association for Computational Linguistics},
  \bibinfo{address}{Beijing, China}, \bibinfo{pages}{666--675}.
\newblock
\urldef\tempurl%
\url{https://doi.org/10.3115/v1/P15-1065}
\showDOI{\tempurl}


\bibitem[\protect\citeauthoryear{Li, Chan, and Maier}{Li et~al\mbox{.}}{2015}]%
        {Maier:VLDB:2015}
\bibfield{author}{\bibinfo{person}{Hao Li}, \bibinfo{person}{Chee~Yong Chan},
  {and} \bibinfo{person}{David Maier}.} \bibinfo{year}{2015}\natexlab{}.
\newblock \showarticletitle{Query From Examples: An Iterative, Data-Driven
  Approach to Query Construction}.
\newblock \bibinfo{journal}{\emph{PVLDB}}  \bibinfo{volume}{8}
  (\bibinfo{year}{2015}), \bibinfo{pages}{2158--2169}.
\newblock


\bibitem[\protect\citeauthoryear{Mihalkova and Mooney}{Mihalkova and
  Mooney}{2007}]%
        {Mihalkova:ICML:07}
\bibfield{author}{\bibinfo{person}{Lilyana Mihalkova} {and}
  \bibinfo{person}{Raymond~J. Mooney}.} \bibinfo{year}{2007}\natexlab{}.
\newblock \showarticletitle{Bottom-up learning of {Markov} logic network
  structure}. In \bibinfo{booktitle}{\emph{ICML}}.
\newblock


\bibitem[\protect\citeauthoryear{Muggleton}{Muggleton}{1995}]%
        {progol}
\bibfield{author}{\bibinfo{person}{Stephen Muggleton}.}
  \bibinfo{year}{1995}\natexlab{}.
\newblock \showarticletitle{Inverse entailment and {Progol}}.
\newblock \bibinfo{journal}{\emph{New Generation Computing}}
  \bibinfo{volume}{13} (\bibinfo{year}{1995}), \bibinfo{pages}{245--286}.
\newblock


\bibitem[\protect\citeauthoryear{Muggleton and Feng}{Muggleton and
  Feng}{1990}]%
        {golem}
\bibfield{author}{\bibinfo{person}{Stephen Muggleton} {and}
  \bibinfo{person}{Cao Feng}.} \bibinfo{year}{1990}\natexlab{}.
\newblock \showarticletitle{Efficient Induction of Logic Programs}. In
  \bibinfo{booktitle}{\emph{ALT}}.
\newblock


\bibitem[\protect\citeauthoryear{Muggleton, Santos, and
  Tamaddoni-Nezhad}{Muggleton et~al\mbox{.}}{2009}]%
        {progolem}
\bibfield{author}{\bibinfo{person}{Stephen Muggleton}, \bibinfo{person}{Jose
  Santos}, {and} \bibinfo{person}{Alireza Tamaddoni-Nezhad}.}
  \bibinfo{year}{2009}\natexlab{}.
\newblock \showarticletitle{{ProGolem}: A System Based on Relative Minimal
  Generalisation}. In \bibinfo{booktitle}{\emph{ILP}}.
\newblock


\bibitem[\protect\citeauthoryear{Picado, Termehchy, and Fern}{Picado
  et~al\mbox{.}}{2017}]%
        {castor:SIGMOD17}
\bibfield{author}{\bibinfo{person}{Jose Picado}, \bibinfo{person}{Arash
  Termehchy}, {and} \bibinfo{person}{Alan Fern}.}
  \bibinfo{year}{2017}\natexlab{}.
\newblock \showarticletitle{Schema Independent Relational Learning}. In
  \bibinfo{booktitle}{\emph{SIGMOD Conference}}.
\newblock


\bibitem[\protect\citeauthoryear{Quinlan}{Quinlan}{1990}]%
        {Quinlan:FOIL}
\bibfield{author}{\bibinfo{person}{J.~Ross Quinlan}.}
  \bibinfo{year}{1990}\natexlab{}.
\newblock \showarticletitle{Learning Logical Definitions from Relations}.
\newblock \bibinfo{journal}{\emph{Machine Learning}}  \bibinfo{volume}{5}
  (\bibinfo{year}{1990}), \bibinfo{pages}{239--266}.
\newblock


\bibitem[\protect\citeauthoryear{Raedt, Poole, Kersting, and Natarajan}{Raedt
  et~al\mbox{.}}{2017}]%
        {SRLNIPS}
\bibfield{author}{\bibinfo{person}{Luc~De Raedt}, \bibinfo{person}{David
  Poole}, \bibinfo{person}{Kristian Kersting}, {and} \bibinfo{person}{Sriraam
  Natarajan}.} \bibinfo{year}{2017}\natexlab{}.
\newblock \showarticletitle{Statistical Relational Artificial Intelligence:
  Logic, Probability and Computation}. In \bibinfo{booktitle}{\emph{NeurIPS}}.
\newblock


\bibitem[\protect\citeauthoryear{Rekatsinas, Chu, Ilyas, and R{\'e}}{Rekatsinas
  et~al\mbox{.}}{2017}]%
        {Rekatsinas2017HoloCleanHD}
\bibfield{author}{\bibinfo{person}{Theodoros~I. Rekatsinas},
  \bibinfo{person}{Xu Chu}, \bibinfo{person}{Ihab~F. Ilyas}, {and}
  \bibinfo{person}{Christopher R{\'e}}.} \bibinfo{year}{2017}\natexlab{}.
\newblock \showarticletitle{{HoloClean}: Holistic Data Repairs with
  Probabilistic Inference}.
\newblock \bibinfo{journal}{\emph{PVLDB}}  \bibinfo{volume}{10}
  (\bibinfo{year}{2017}), \bibinfo{pages}{1190--1201}.
\newblock


\bibitem[\protect\citeauthoryear{Richardson and Domingos}{Richardson and
  Domingos}{2006}]%
        {Richardson:2006:MLN}
\bibfield{author}{\bibinfo{person}{Matthew Richardson} {and}
  \bibinfo{person}{Pedro~M. Domingos}.} \bibinfo{year}{2006}\natexlab{}.
\newblock \showarticletitle{Markov logic networks}.
\newblock \bibinfo{journal}{\emph{Machine Learning}}  \bibinfo{volume}{62}
  (\bibinfo{year}{2006}), \bibinfo{pages}{107--136}.
\newblock


\bibitem[\protect\citeauthoryear{Weis, Naumann, Jehle, Lufter, and
  Schuster}{Weis et~al\mbox{.}}{2008}]%
        {Weis:2008:IDD:1454159.1454165}
\bibfield{author}{\bibinfo{person}{Melanie Weis}, \bibinfo{person}{Felix
  Naumann}, \bibinfo{person}{Ulrich Jehle}, \bibinfo{person}{Jens Lufter},
  {and} \bibinfo{person}{Holger Schuster}.} \bibinfo{year}{2008}\natexlab{}.
\newblock \showarticletitle{Industry-scale duplicate detection}.
\newblock \bibinfo{journal}{\emph{PVLDB}}  \bibinfo{volume}{1}
  (\bibinfo{year}{2008}), \bibinfo{pages}{1253--1264}.
\newblock


\bibitem[\protect\citeauthoryear{Wijsen}{Wijsen}{2003}]%
        {10.5555/645505.656435}
\bibfield{author}{\bibinfo{person}{Jef Wijsen}.}
  \bibinfo{year}{2003}\natexlab{}.
\newblock \showarticletitle{Condensed Representation of Database Repairs for
  Consistent Query Answering}. In \bibinfo{booktitle}{\emph{Proceedings of the
  9th International Conference on Database Theory}}.
  \bibinfo{publisher}{Springer-Verlag}, \bibinfo{address}{Berlin, Heidelberg},
  \bibinfo{pages}{378?393}.
\newblock
\showISBNx{3540003231}


\bibitem[\protect\citeauthoryear{Yakout, Elmagarmid, Neville, Ouzzani, and
  Ilyas}{Yakout et~al\mbox{.}}{2011}]%
        {10.14778/1952376.1952378}
\bibfield{author}{\bibinfo{person}{Mohamed Yakout}, \bibinfo{person}{Ahmed~K.
  Elmagarmid}, \bibinfo{person}{Jennifer Neville}, \bibinfo{person}{Mourad
  Ouzzani}, {and} \bibinfo{person}{Ihab~F. Ilyas}.}
  \bibinfo{year}{2011}\natexlab{}.
\newblock \showarticletitle{Guided Data Repair}.
\newblock \bibinfo{journal}{\emph{Proc. VLDB Endow.}} \bibinfo{volume}{4},
  \bibinfo{number}{5} (\bibinfo{date}{Feb.} \bibinfo{year}{2011}),
  \bibinfo{pages}{279?289}.
\newblock
\showISSN{2150-8097}
\urldef\tempurl%
\url{https://doi.org/10.14778/1952376.1952378}
\showDOI{\tempurl}


\bibitem[\protect\citeauthoryear{Zeng, Patel, and Page}{Zeng
  et~al\mbox{.}}{2014}]%
        {QuickFOIL}
\bibfield{author}{\bibinfo{person}{Qiang Zeng}, \bibinfo{person}{Jignesh~M.
  Patel}, {and} \bibinfo{person}{David Page}.} \bibinfo{year}{2014}\natexlab{}.
\newblock \showarticletitle{{QuickFOIL}: Scalable Inductive Logic Programming}.
\newblock \bibinfo{journal}{\emph{PVLDB}}  \bibinfo{volume}{8}
  (\bibinfo{year}{2014}), \bibinfo{pages}{197--208}.
\newblock


\end{thebibliography}
		\appendix
\section{Omitted proofs}
\label{sec:omitted}
\noindent
{\bf Proof of Proposition 4.3:}\\
We show that the algorithm never adds a repair literal that reverts 
the impact of a previously added repair literal. This may happen only if there is a chain of CFDs, such as,
$\phi:$ $AC \rightarrow B, t_p$ and $\phi':$ $BD \rightarrow A, t'_p$ over relation $R$. 
Consider a violation of $\phi$ with two literals $l_1$ and $l_2$ of $R$ in the input bottom-clause in which 
the variables associated with attributes $A$ and $C$ are equal and match $t_p$ but the 
variables associated with $B$ are not equal or do not satisfy $t_p$.
Let repair literal $r_{\phi}$ unify the value of variables in attribute $B$ for a violation of $\phi$ 
and set them to a constant if their corresponding element in $t_p$ is a constant.
We show that the repair introduced by $r_{\phi}$ does not cause a violation of $\phi'$ for literals
$l_2$ and $l_1$. Let the values for $A$ in $t_p$ and $t'_p$ be equal or at least one of them is '-'. 
If the variables assigned to attribute $D$ in $l_2$ and $l_1$ are equal 
and match $t'_p$, $l_2$ and $l_1$ satisfy $\phi'$ as the violation of $\phi$ indicate that 
the variables assigned to attribute $A$ in $l_2$ and $l_1$ are equal.
If the variables assigned to attribute $D$ are not equal, $l_2$ and $l_1$ satisfy $\phi'$.
Now, assume that values for $A$ in $t_p$ and $t'_p$ are unequal constants. Then, $\phi$ and $\phi'$ are not
consistent. The proposition is proved for longer chains similarly.

Let $C^{s}_{e}$ be a repaired clause created by the application of a set of repair literals $r$ 
in $C_e$. Application of repair literals in $r$ correspond to applying some MDs in $\Sigma$ 
or CFDs in $\Phi$ on $I_e$ that creates a repair of $I_e$ such that $C^{s}_{e}$ covers $e$. 
Thus, according to Definition 3.4, $C_e$ covers $e$. \\

\noindent
{\bf Proof for Theorem 4.6:}\\
	Let $\theta$ be the substitution mapping from $C$ to $D$. 
	Let $V^D$ be a mapped repair literal in $D$ with corresponding literal in $C$ $V^C$ such that 
	$V^C \theta= V^D$. If variables/constants $x$ and $y$ are equal or similar in $C$, 
	$\theta(x)$ and $\theta(y)$ are also equal in $D$.
	Thus, $V^C$ and $V^D$ are applied to literals $L^{C}_1$ and $L^{C}_2$ and $L^{D}_1$ and $L^{D}_2$
	such that $L^{C}_i \theta=L^{D}_i$, $1 \leq i \leq 2$, and modify the variables in those 
	literals that are mapped via $\theta$.
	In our clauses, there is not any repair literal $V$ whose condition is false. Otherwise, $V$
	either has not been placed in the clause or has been removed from it after applying another repair 
	literal that makes the condition of $V$ false.
	Let $M^{C}_i$ and $M^{D}_i$ be the modifications of $L^{C}_i$ and $L^{D}_i$, respectively. 
	By applying $V^C$ and $V^D$ on their corresponding literal, we will have
	$M^{C}_i \theta=M^{D}_i$. Moreover, if $V^C$ replaces a variable $x$ with 
	$v_x$, $V^C$ will also replace $y = x \theta$ with $v_y$ such that $v_y = v_x \theta$. 
	Each repair literal either unifies two variables in two non-repair literals using fresh (MD applications)
	or existing variables (CFD repair by modifying the right-hand side) or replaces existing variable(s) with  fresh ones (CFD repair by changing the left-hand side).
	Since equal variables in $C$ are mapped to equal ones in $D$, the applications of $V^C$ and $V^D$ will result
	in removing repair literals $U^C$ and in $U^D$ such that $U^C \theta U^D$.

	The unmapped repair literals in $D$ do not modify $\theta$ for the variables as they are not 
	connected to the mapped non-repair literals of $D$.
	Let $C_r$ and $D_r$ be result of applying $V^C$ and $V^D$, respectively.
	There is a subsumption mapping between $C_r$ and $D_r$ using a substitution $\theta' \subseteq \theta$.
	$\theta'$ may not have some of the variables that exist in $C$ and $D$ but not $C_r$ and $D_r$.
	Thus, there is a $theta$-subsumption between the clauses after each repair. We repeat the same argument 
	for applications of each repair literals other than $V^C$ and $V^D$ in $C$ and $D$ 
	and also every repair literal in every resulting repairs of $C$ and $D$, such as $C_r$ and $D_r$.
	As there are $\theta$-subsumption between every repair of $C$ and some repair of $D$, 
	according to the definition of logical entailment for clauses with repair literals, $C \models D$.\\

\noindent
{\bf Proof for Proposition 4.8:}\\
Since we drop each literal in a clause with its repair literals, it 
corresponds to dropping the repairs of this literal in each repaired version of the clause during its 
generalization over its corresponding clean database. Thus, according to 
Theorem 4.6, the proof is similar to the one of Theorem 3 in \cite{progolem}.\\

\noindent
{\bf Proof for Theorem 4.9:}\\
	Let $\mC$ and $\mD$ be the set of repaired clauses for 
	$C$ and $D$, respectively. According to definition of logical entailment for clauses with
	repair literals, for each $C_r \in \mC$, there is a $D_r in \mD$ such that $C_r \models D_r$.
	Since $C_r$ and $D_r$ do not have any repair literal, there is a substitution mapping $\theta_r$ such that
	$C_r \theta_1 \subseteq D_r$. 
	Let $C_o$ and $D_o$ be clauses where the application of a single 
	repair literal result in producing $C_r$ and $D_r$. 
	We show that for each $C_o$ there is a $D_o$ such that there is a substitution mapping $\theta_o$ between $C_o$ and $D_o$.
	$C_r$ and $D_r$ are defined over the same database with the same set of constraints
	and we add similarity literal(s) to clauses during the process of (ground) bottom-clause construction if 
	they satisfy the application of an MD. 
	For every pair of literals $L_1$ and $L_2$ in $C_r$, if they satisfy the left-hand side of an MD, 
	there are repair literals to apply the MD in $C$ for these literals. The same is true in $D_r$ and $D$. 
	If the repair applied on $C_o$ is due to an MD, as $\theta_r$ preserves 
	similarity and equality between variables, there is a $D_o$ for $D_r$ such that 
	the repair applied on $D_o$ must also be according to an MD. 
	Also, we do not have a CFD and MD that share their left-hand side as MDs are defined over distinct relations.
	Let $V_j^C(x_j,v_{x_j})$, $V_j^D(y_j,v_{y_j})$, $1 \leq j \leq 2$
	be the repair literals that modify $C_o$ and $D_o$ to $C_r$ and $D_r$.
	The mapping $\theta_r$ maps $v_{x_j}$ to $v_{y_j}$. If variables $x_j$ and $y_j$ do not appear in 
	$C_r$ and $D_r$, we add new mappings from $x_j$ to $y_j$ to $theta_r$ to get subsumption mapping $theta_o$
	between $C_o$ and $D_o$.
	Otherwise, $x_1$ and $x_2$ and $y_1$ and $y_2$ appear in similarity literals in $C$ and $D$, respectively. 
	Thus, they are mapped using $theta_r$. Thus, there is a subsumption mapping 
	between $C_o$ and $D_o$ in both cases.\\

\noindent
{\bf Proof for Proposition 4.10:}\\ 
According to Theorem 4.6, if $s_C$ $\theta$-subsumes $G_{e^{-}}$, $C$ covers $e^{-}$ relative to $I$ based on Definition 3.6.\\

\noindent
{\bf Proof for Theorem 4.11:}\\
	Without loss of generality, assume that all learned definitions contain one clause.
	Let ${\bf J}  = \mathit{RepairedInst}(I, \Sigma, \Phi)  = \{ J_1, \ldots, J_n\}$.
	We show that $\mathit{BC}(I, e, \Sigma, \Phi) = C$ is a compact representation of 
	$\mathit{BC}_{r}({\bf J}, e) = \{ C_1, \ldots, C_n \}$.
	
	Let $\mathit{RepairedCls}(C)$ $=$ $\{ C'_1, \ldots, C'_m \}$. 
	We remove the literals that are not head-connected in each clause in $\{ C'_1, \ldots, C'_m \}$.
	Let $\{ J^{C'_1}, \ldots, J^{C'_m} \}$ be the canonical database instances of $\{ C'_1, \ldots, C'_m \}$ \cite{AliceBook}. 
	This set is the same set as the one generated by applying $\mathit{RepairedInst}(I^{C}, \Sigma, \Phi)$, 
	where $I^{C}$ is the canonical database instance of $C$.
	
	Let $\{ J^{C_1}, \ldots, J^{C_n} \}$ be the canonical database instances of $\{ C_1, \ldots, C_n \}$.
	By definition, $I^{C}$ contains all tuples that are related to $e$, either by exact or similarity matching (according to MDs in $\Sigma$). Because $\mathit{RepairedInst}(I^{C}, \Sigma) = \{ J^{C'_1}, \ldots, J^{C'_m} \}$,
	all tuples that may appear in an instance in $\{ J^{C_1}, \ldots, J^{C_n} \}$ must also appear in an instance in $\{ J^{C'_1},$ $\ldots,$ $J^{C'_m} \}$.
	
	A tuple $t$ may appear in an instance in $J^{C'_j} \in \{ J^{C'_1},$ $\ldots,$ $J^{C'_m} \}$, but not appear in the corresponding instance $J^{C_i} \in \{ J^{C_1}, \ldots, J^{C_n} \}$.
	In this case, $t$ became disconnected from training example $e$ when generating the repair $J_i$, which is a superset of $J^{C_i}$. Then, when building bottom-clause $C_i$ from $J_i$, a literal was not created for $t$.
	However, the same tuple would also become disconnected from training example $e$ in $J^{C'_j}$.
	Because we remove literals that are not head-connected in each clause in $\{ C'_1, \ldots, C'_m \}$, we would remove $t$ from $C'_j$.
	
	Let $R(\bar{x}, y, \bar{z})$ and $R(\bar{x}, y', \bar{z'})$ be a violation of CFD $(X \rightarrow A, t_{p})$
	in the bottom-clause generated by our algorithm. Each constant or variables in these literal 
	remains unchanged in at least one application of the repair literals. 
	Thus, if these literals are connected to the positive example $e$ in at least one of the 
	repairs of $I$ they will appear in the generated bottom-clause of our algorithm. Also, if they appear in the 
	our produced bottom-clause, they must appear at least in one of the repairs of $I$.
	
	The sets of canonical database instances $\{ J^{C'_1}, \ldots, J^{C'_m} \}$ and $\{ J^{C_1}, \ldots, J^{C_n} \}$ are both generated using the function {\it RepairedInst} with the same 
	dependencies $\Sigma$ and $\Phi$, and contain only tuples related to $e$.
	Therefore, $\mathit{}(C) = \{ C'_1, \ldots, C'_m \}$ is equal to $\{ C_1, \ldots, C_n \}$.\\

\noindent
{\bf Proof for Theorem 4.12:}\\
	We prove the theorem for the MDs. The proof for CFDs is done similarly.
	Let ${\bf J}  = \mathit{RepairedInst}(I, \Sigma)$.
	We show that the clause
	$\mathit{Generalize}(C, I, e', \Sigma)$ $=$ $C^*$ is a compact 
	representation of $\mathit{Generalize}_{r}({\bf C}, {\bf J}, e') = \{ C^*_1, \ldots, C^*_n \}$, 
	i.e.\\
	$\mathit{RepairedCls}(C^*) = \{ C^*_1, \ldots, C^*_n \}$.
	
	Assume that the schema is $\mR$ $=$ $\{ R_1(A, B),$ $R_2(B, C) \}$ 
	and we have MD $\phi: R_1[B] \approx R_2[B] \rightarrow R_1[B] \rightleftharpoons R_2[B]$.
	This proof generalizes to more complex schemas.
	Assume that database instance $I$ contains tuples $R_1(a,b)$, $R_2(b', c)$, and $R_2(b'', c)$, and that $b \approx b'$ and $b \approx b''$.
	Then, bottom-clause $C$ has the form
	\begin{align*}
		T(u) \leftarrow& L'_1, \ldots, L'_{l-1}, \\
		& R_1(a,b), R_2(b', c), V(b,x_b), V(b',x_{b'}), x_b=x_{b'}\\
		& R_2(b'', d), V(b,y_b), V(b'',y_{b''}), y_b=y_{b''}\\
		& L'_{l}, \ldots, L'_n,
	\end{align*}
	where $L'_k$, $1 \le k \le n$, is a literal.
	
	Now consider two stable instances generated by $\mathit{RepairedInst}(I, \Sigma)$: $J_1$, 
	which contains tuples $R_1(a, x_{b}),$ $R_2(x_{b}, c),$ $R_2(b'', c)$; and $J_2$, 
	which contains tuples 
	$R_1(a, y_{b}),$ $R_2(b', c),$ $R_2(y_{b}, c)$.
	The bottom-clause $C_1$ over instance $J_1$ has the form
	\begin{align*}
		T(u) \leftarrow& L_1, \ldots, L_{l-1}, \\
		& R_1(a, x_{b}), R_2(x_{b}, c), R_2(b'', c),\\
		& L_{l}, \ldots, L_n,
	\end{align*}
	and the bottom-clause $C_2$ over instance $J_2$ has the form
	\begin{align*}
		T(u) \leftarrow& L_1, \ldots, L_{l-1}, \\
		& R_1(a, y_{b}), R_2(b', c), R_2(y_{b}, c),\\
		& L_{l}, \ldots, L_n,
	\end{align*}
	where $L_k$, $1 \le k \le n$, is a literal. 
	
	We want to generalize $C_1$ to cover another training example $e'$. 
	Let $G_{e'}$ be the ground bottom-clause for $e'$ and $G'_{e'}$ be a repaired clause of $G_{e'}$.
	The literals in $C_1$ that are blocking will depend on the content of the ground bottom-clause $G'_{e'}$. 
	Assume that the sets of literals $\{ L'_1, \ldots, L'_n\}$ in clause $C$ and the set of literals $\{ L_1, \ldots, L_n\}$ in clauses $C_1$ and $C_2$ are equal.
	We consider the following cases for the literals that are not equal.
	The same cases apply when we want to generalize any other clause generated from a repaired instance, e.g., $C_2$.
	
	Case 1: $G'_{e'}$ contains the literals  $R_1(a, x_{b})$ and $R_2(x_{b}, c)$.
	In this case, $R_1(a, x_{b})$ and $R_2(x_{b}, c)$ are {\it not} blocking literals, i.e., they are not removed from $C_1$.
	$G_{e'}$ also contains literals $R_1(a, b), R_2(b', c), V(b,x_b),$ $V(b',x_{b'}), x_b=x_{b'}$. 
	Therefore, the same literals are {\it not} blocking literals in $C$ either.
	
	Case 2: $G'_{e'}$ contains literals with same relation names but not the same pattern.
	Assume that $G'_{e'}$ contains the literals $R_1(a, b)$ and $R_2(d, c)$, i.e., they do not join.
	In this case, literal $R_2(x_{b}, c)$ in $C_1$ is a blocking literal because it joins with a literal that appears previously in the clause, $R_1(a, x_{b})$. Hence, it is removed.
	$G_{e'}$ also contains literals $R_1(a, b)$ and $R_2(d, c)$. 
	Because in clause $G'_{e'}$, created from the repaired instance, 
	these literals do not join, in $G_{e'}$ they do not join either. In this case, the blocking literals in $C$ are $V(b,x_b), V(b',x_{b'}),$ $x_b=x_{b'}, R_2(b', c)$.
	
	Case 3: $G'_{e'}$ contains $R_1(a, x_{b})$, but not $R_2(x_{b}, c)$.
	In this case, literal $R_2(x_{b}, c)$ is a blocking literal in $C_1$. 
	Therefore, it is removed.
	$G_{e'}$ also contains literals $R_1(a, b)$ and $V(b,x_b), V(b',x_{b'}),$ $x_b=x_{b'}$, 
	but not $R_2(b', c)$. Therefore, literal $R_2(b', c)$ in $C$ is also blocking and it is removed.
	
	Case 4: $G'_{e'}$ contains $R_2(x_{b}, c)$, but not $R_1(a, x_{b'})$.
	This case is similar to the previous case.
	
	Case 5: $G'_{e'}$ contains neither $R_1(a, x_{b})$ nor $R_2(x_{b}, c)$.
	In this case, both $R_1(a, x_{b})$ and $R_2(x_{b}, c)$ are blocking; hence they are removed.
	$G_{e'}$ does not contain literals $R_1(a, b)$, $R_2(b', c)$, $V(b,x_b), V(b',x_{b'}),$ $x_b=x_{b'}$. 
	Hence, these literals are also blocking literals in $C$ and are removed.
	
	The generalization operations $\mathit{Generalize}(C, I, e', \Sigma)$ and $\mathit{Generalize}_{r}({\bf C}, {\bf J}, e')$ consist of removing blocking literals from $C$ and ${\bf C}$ respectively. 
	We have shown that the same literals are blocking over both the clauses. Therefore, 
	$\mathit{RepairedCls}(C^*) = \{ C^*_1, \ldots, C^*_n \}$.\\

	\end{document}